\documentclass[aps,pra,twocolumn,superscriptaddress,floatfix,nofootinbib,showpacs,longbibliography,groupedaddress]{revtex4-1}

\usepackage[utf8]{inputenc}  
\usepackage[T1]{fontenc}     
\usepackage[british]{babel}  
\usepackage[sc,osf]{mathpazo}
\usepackage[scaled=0.86]{berasans}  
\usepackage[colorlinks=true, citecolor=blue, urlcolor=blue]{hyperref}  
\usepackage{graphicx} 
\usepackage[babel]{microtype}  
\usepackage{amsmath,amssymb,amsthm,bm,amsfonts,mathrsfs,bbm} 

\usepackage{xspace}  
\usepackage{pgf,tikz}
\usepackage{xcolor}
\usepackage{multirow}
\usepackage{array}
\usepackage{bigstrut}
\usepackage{braket}
\usepackage{color}
\usepackage{natbib}
\usepackage{multirow}
\usepackage{mathtools}
\usepackage{float}
\usepackage[caption = false]{subfig}
\usepackage{xcolor,colortbl}
\usepackage{color}
\usepackage{booktabs}
\usepackage{scalerel}

\newcommand{\be}{\begin{equation}}
\newcommand{\ee}{\end{equation}}
\newcommand{\ba}{\begin{eqnarray}}
\newcommand{\ea}{\end{eqnarray}}
\newcommand{\ketbra}[2]{|#1\rangle \langle #2|}
\newcommand{\tr}{\operatorname{Tr}}

\newtheorem{theorem}{Theorem}

\newtheorem{definition}{Definition}

\newtheorem{observation}{Observation}

\def\>{\rangle}
\def\<{\langle}

\usepackage{centernot}
\usepackage{subfig}

\begin{document}

\title{Mutually Unbiased Balanced Functions \& Generalized Random Access Codes}

\author{Vaisakh M}
\affiliation{School of Physics, IISER Thiruvananthapuram, Vithura, Kerala 695551, India.}

\author{Ram krishna Patra}
\affiliation{School of Physics, IISER Thiruvananthapuram, Vithura, Kerala 695551, India.}

\author{Mukta Janpandit}
\affiliation{School of Physics, IISER Thiruvananthapuram, Vithura, Kerala 695551, India.}

\author{Samrat Sen}
\affiliation{School of Physics, IISER Thiruvananthapuram, Vithura, Kerala 695551, India.}

\author{Anubhav Chaturvedi}
\email{anubhav.chaturvedi@research.iiit.ac.in}
\affiliation{Institute of Theoretical Physics and Astrophysics, National Quantum Information Centre, Faculty of Mathematics, Physics and Informatics, University of Gdansk, 80-952 Gdansk, Poland}
\affiliation{International Centre for Theory of Quantum Technologies, University of Gdansk, 80-308, Gdansk, Poland}

\author{Manik Banik}
\email{manik11ju@gmail.com}
\affiliation{School of Physics, IISER Thiruvananthapuram, Vithura, Kerala 695551, India.}

\begin{abstract}
Quantum resources and protocols are known to outperform their classical counterparts in variety of communication and information processing tasks. Random Access Codes (RACs) are one such cryptographically significant family of bipartite communication tasks, wherein, the sender encodes a data set (typically a string of input bits) onto a physical system of bounded dimension and transmits it to the receiver, who then attempts to guess a randomly chosen part of the sender's data set (typically one of the sender's input bits). In this work, we introduce a generalization of this task wherein the receiver, in addition to the individual input bits, aims to retrieve randomly chosen functions of sender's input string. Specifically, we employ sets of \textit{mutually unbiased balanced functions} (MUBS), such that perfect knowledge of any one of the constituent functions yields no knowledge about the others. We investigate and bound the performance of (i.) classical, (ii.) quantum prepare and measure, and (iii.) entanglement assisted classical communication (EACC) protocols for the resultant generalized RACs (GRACs). Finally, we detail the case of GRACs with three input bits, find maximal success probabilities for classical, quantum and EACC protocols, along with the effect of noisy quantum channels on the performance of quantum protocols. Moreover, with the help of this case study, we reveal several characteristic properties of the GRACs which deviate from the standard RACs.

\end{abstract}


\maketitle	
\section{Introduction}
Quantum information theory entails the study of quantum resources and protocols which are known to enable a plethora of communication and information processing tasks, which otherwise remain unattainable by their classical counterparts governed by Shannon's information theory \cite{Shannon48}. For instance, in quantum super-dense coding, a sender (say Alice) can transfer two classical bits of information to a distant receiver (say Bob) by transmitting a single two-level quantum system, with the aid of pre-shared entanglement \cite{Bennett92}. Similarly, the counter-intuitive feature of quantum entanglement is known to empower several seemingly impossible tasks. However, in absence of entanglement, the utility of quantum systems in communication tasks is constrained by certain fundamental no-go results. For instance, the Holevo theorem \cite{Holevo73} constrains the informational utility of individual quantum systems. Specifically, no more than $n$ classical bits of information can be reliably transmitted using $n$ quantum bits. Recently, a more stringent constraint on quantum communication was established by Frenkel \& Weiner, namely, it has been established that classical information storage capacity of a $d$-level quantum system is same as that of a classical $d$-state system \cite{Frenkel15}. 

These no-go results seem to point towards the conclusion that in absence of entanglement, quantum resources and protocols might not be better than their classical counterparts for transmission of classical information. However, in actuality, even without entanglement, finite dimensional quantum systems can outperform their classical counterparts in a large variety of stochastic communication tasks. $(n\to1)$ Random Access Codes (RACs) constitute such a class of communication tasks wherein the sender is tasked with encoding a string of $n$ bits onto a single bit of message, such that the receiver can decode any one of the randomly chosen initial bits with certain probability of success. It is known that if the message is encoded onto a qubit \footnote{This problem was first studied by Stephen Wiesner under the name \textit{conjugate coding} \cite{Wiesner83}}, the parties can attain higher success probability than any classical strategy entailing a bit of communication \cite{Ambainis99,Ambainis02}. RACs utilizing quantum resources (often referred to as QRACs) have a plethora of applications, for instance, in connection with quantum communication complexity (see \cite{Buhrman10} and references therein), network coding and locally decodable codes \cite{Klauck01,Kerenidis04,Aaronso04,Wehner05,Gavinsky06,Hayashi07}. Moreover, RACs have found a number of foundational implications \cite{Spekkens09,Banik15,Chailloux16,Hameedi17,Ghorai18,Saha19,Ambainis19,Chaturvedi20}, in particular, pre-shared Entanglement Assisted Random Access Codes (EARACs) are closely related to non-local games \cite{Pawloski10,AnubhavRAC}, and form the basis for the principle of Information Causality \cite{Pawloski09,Safi11}. 

Finally, RACs form a cryptographic primitive, and consequently, form the basis of Quantum Key Distribution (QKD) schemes \cite{Marcin2011,Chaturvedi2018,Anubhav2020}. One of the features of $(n \to 1)$ RACs which makes them a suitable crytographic primitive is that in each round the receiver intends to retrieve a single bit of the sender's $n$ bit data, and as these bits are independently distributed, such an exclusive decoding reveals no non-trivial information about the other bits. 
In this work, we propose a generalization of the RAC task, referred to as GRACs, which expands on this property. To this end, we introduce sets of $n\mapsto 1$ bit functions that are \textit{mutually unbiased} (called MUBS) such that perfect knowledge of any one of the constituent functions in a MUBS, yields no information about the others. For each such MUBS, the communication task wherein, in each round the receiver intends to decode a constituent function, forms a GRAC.

The manuscript is organized in the following manner: in Section \ref{sec2} we formally introduce the concept of mutually unbiased balanced functions and their sets; in Section \ref{sec3} we provide the definition of generalized RAC (GRAC) task, describe (i.) classical, (ii.) quantum prepare and measure, and (iii.) entanglement assisted classical communication (EACC) protocols for GRACs, and derive general bounds on the success probabilities of these protocols in GRACs. In Section \ref{sec4} we provide a detailed study of the case when the sender receives three idependently distributed bits as input. Here, we find that majority-encoding-identity-decoding may not be one of the optimal classical strategies for GRAC which is always the case for RACs. We also report an interesting feature of classical GRACs, namely, the maximal average classical success probability of a GRAC gets increased with addition of new function(s) to the already existing one, a phenomenon referred to as, `{\it the harder the goal, the greater the payoff}'. In the same section, we also analyse the performance quantum prepare and measure protocols, the effect noisy channel on their performance, and the performance of EACC protocols. In section \ref{sec5} we present a brief discussion along with some relevant open questions for further research.

\section{Mutually Unbiased Balanced Functions}\label{sec2}
This section specifies the definitions of balanced and mutually unbiased balanced functions and sets of mutually unbiased balanced functions, providing key examples along with some notational preliminaries for subsequent use throughout the rest of the manuscript.
\begin{definition} [Balanced Functions] A Boolean function $f:\{0,1\}^n\to\{0,1\}$ is balanced if its outputs yield as many $0$s as $1$s over its input set.
\end{definition}
In this work, we consider $n\to1$ bit Boolean functions which take as input an $n$ bit string $\mathbf{x}\equiv \{x_1x_2\cdots x_n\} \in \{0,1\}^n$ to produce a bit of output, i.e., $f:\{0,1\}^n\to\{0,1\}$. Consequently, for such functions to be balanced, the cardinality of the set of inputs mapped to $0$, $\mathcal{X}_{f(\mathbf{x})=0}\equiv\{\mathbf{x}\in\{0,1\}^n~|~f(\mathbf{x})=0\}$, must be the same as the cardinality of the set of inputs mapped to $1$, $\mathcal{X}_{f(\mathbf{x})=1}\equiv\{\mathbf{x}\in\{0,1\}^n~|~f(\mathbf{x})=1\}$, i.e., $|\mathcal{X}_{f(\mathbf{x})=0}| = |\mathcal{X}_{f(\mathbf{x})=1}| = 2^{n-1}$.
Furthermore, this implies that for such balanced functions, and a uniformly distributed string of inputs, the probability of obtaining a $0$ is the same as the probability of obtaining a $1$, i.e., $\forall \mathbf{x}\in \{0,1\}^{n}: p(\mathbf{x})=\frac{1}{2^n}$,
$p(f(\mathbf{x})=0)=p(f(\mathbf{x})=1)=\frac{1}{2}$. Next, we introduce the notion of \textit{mutual unbiasedness} for balanced functions. 
\begin{definition} [MUBF] A pair of balanced Boolean functions $f_1,f_2$ is called mutually unbiased if exactly half of the inputs yielding an output for one function yields the same output for the other function.
\end{definition}
In particular, for a pair of MUBFs $f_1,f_2:\{0,1\}^n\to\{0,1\}$, the sets $\mathcal{X}_{f_1(\mathbf{x})=0},\mathcal{X}_{f_1(\mathbf{x})=1}$ have equal overlaps with the sets $\mathcal{X}_{f_2(\mathbf{x})=0},\mathcal{X}_{f_2(\mathbf{x})=1}$, i.e., $\forall i,j \in\{0,1\}: |\mathcal{X}_{f_1(\mathbf{x})=i} \cap \mathcal{X}_{f_2(\mathbf{x})=j}|=2^{n-2}$. This further implies for a uniformly distributed string of inputs $\forall i,j,k,l \in \{0,1\}:p(f_2(\mathbf{x})=i|f_1(\mathbf{x})=j)=p(f_1(\mathbf{x})=k|f_2(\mathbf{x})=l)=\frac{1}{2}$. Next, we define sets of mutually unbiased balanced functions.
\begin{definition} [MUBS]
A set of functions $\mathcal{F}\equiv \{f_i\}^{|\mathcal{F}|}_{i=1}$ forms a mutually unbiased balanced set if all of the constituent functions are balanced and pairwise mutually unbiased.  
\end{definition}
For a $n$-bit string input, consider the set of $2^{n}-1$ functions $\mathcal{F}^n_\mathcal{R}\equiv\{f_{\mathbf{r}}(\mathbf{x})=\bigoplus^n_{i=1}r_ix_i|\mathbf{r}\equiv\{r_1,\ldots,r_n\}\in \mathcal{R}\}$ where $\mathcal{R}\equiv\{\mathbf{r}\in \{0,1\}^n|\sum_ir_i\geq 1\}$. Notice that all functions in the set are balanced, i.e., $\forall f \in \mathcal{F}^n_R: |\mathcal{X}_{f(\mathbf{x})=0}|=|\mathcal{X}_{f(\mathbf{x})=1}|=2^{n-1}$. Now any two distinct functions $f_i,f_j\in \mathcal{F}^n_{\mathcal{R}}$ differ by XOR of alteast one completely independent bit, and XOR with a random bit obscures all information (one-time pad ), all functions in such a set are pairwise mutually unbiased, deeming the set to be MUBS. For instance, for the simplest case of two bit input functions the set $\mathcal{F}^2_{\mathcal{R}}=\{x_1,x_2,x_1\oplus x_2\}$ forms a MUBS. Similarly, for the case of three bit input functions the set $\mathcal{F}^3_{\mathcal{R}}=\{x_1,x_2,x_3,x_1\oplus x_2,x_2\oplus x_3,x_1\oplus x_3,x_1\oplus x_2\oplus x_3\}$ forms a MUBS. In general, it is easy to see that any  non-trivial subset of $\mathcal{F}^n_{\mathcal{R}}$ forms a MUBS, i.e., the sets $\mathcal{F}^n_{\mathcal{R}_j}\equiv\{f_{\mathbf{r}}(\mathbf{x})=\bigoplus^n_{i=1}r_ix_i|\mathbf{r}\equiv\{r_1,\ldots,r_n\}\in \mathcal{R}_j\}$ where $\mathcal{R}_j \subseteq \mathcal{R}$. 

\section{Generalized random access codes}\label{sec3}
In this section, we start by introducing generalized random access codes which utilize mutually unbiased balanced functions defined above.

\begin{definition} [$\mathbf{(n,\mathcal{R}_i)}$-GRAC] An $(n,\mathcal{R}_i)$ generalized random access code (GRAC) is a bipartite one-way communication task, wherein in each round the sender (Alice) receives a uniformly distributed input $n$-bit string $\textbf{x} \in \{0,1\}^n$ which they encode onto a message, which is transmitted to the receiver (Bob). Bob upon receiving the transmission from Alice, decodes the message based on their uniformly distributed input $\mathbf{y}\in \mathcal{R}_i$ where $\mathcal{R}_i \subseteq \mathcal{R} \equiv\{\mathbf{r}\in\{0,1\}^n|\sum_ir_i\geq 1\}$ and produces an output bit $z\in\{0,1\}$. They win a round if $z=f_\mathbf{y}(\mathbf{x})=\bigoplus^n_{i=1}x_iy_i$. They gauge their performance on the basis of their average success probability $s^{(n,\mathcal{R}_i)}=\frac{1}{2^n|\mathcal{R}_i|}\sum_{\mathbf{x},\mathbf{y}}p(z=f_\mathbf{y}(\mathbf{x})|\mathbf{x},\mathbf{y})$. 
\end{definition}
We note here that the standard $(n\rightarrow1)$ random access codes (RACs) form restricted cases of GRACs, specifically when Bob's input $\mathbf{y}$ is uniformly sampled from $\mathcal{R}_{RAC} \equiv \{\mathbf{r}\in\{0,1\}^n|\sum_ir_i=1\}$. We denote the success probability of $(n\rightarrow1)$ RACs by $s^{(n\rightarrow1)}$.

In this work, we study three distinct classes of communication protocols for $(n,\mathcal{R}_i)$ GRAC, 

$(i.)$ A classical communication $\mathcal{C}$ protocol for $(n,\mathcal{R}_i)$ GRAC is one wherein Alice encodes their input string $\mathbf{x}$ onto a bit $\omega \in\{0,1\}$, based on a encoding scheme $\mathcal{E}$ entailing conditional probability distributions of the form $\{p_{\mathcal{E}}(\omega|\mathbf{x})\}$. Bob decodes the message based on their input to produce the output $z$ employing a decoding scheme $\mathcal{D}$ entailing conditional probability distributions of the form $\{p_{\mathcal{D}}(z|\omega,\mathbf{y})
\}$. The average success probability for such a protocol is $s^{(n,\mathcal{R}_i)}_{\mathcal{C}} = \frac{1}{2^n|\mathcal{R}_i|} \sum_{\mathbf{x},\mathbf{y},\omega}p_{\mathcal{E}}(\omega|\mathbf{x})p_{\mathcal{D}}(z=f_\mathbf{y}(\mathbf{x})|\omega,\mathbf{y})$.  The maximal classical success probability of $(n,\mathcal{R}_i)$ GRAC, $S^{(n,\mathcal{R}_i)}_{\mathcal{C}}$ has the expression, 
\small
\begin{equation}
 S^{(n,\mathcal{R}_i)}_{\mathcal{C}}=\max_{\mathcal{E},\mathcal{D}}\left\{\frac{1}{2^n|\mathcal{R}_i|} \sum_{\mathbf{x},\mathbf{y},\omega}p_{\mathcal{E}}(\omega|\mathbf{x})p_{\mathcal{D}}\left(z=f_\mathbf{y}(\mathbf{x})|\omega,\mathbf{y}\right)\right\}.  
\end{equation}
\normalsize
We note here that, as we are considering average success probability, the parties gain no advantage even they have access to an arbitrary amount of shared randomness \cite{Ambainis08}. Moreover, it is straightforward to see that for average success probability it is enough to consider only deterministic encoding and decoding schemes. Consequently, the optimal classical protocols for $(n,\mathcal{R}_i)$ GRAC, without loss of generality, comprise of a deterministic encoding schemes such that $\omega=f_{\mathcal{E}}(\mathbf{x})$, and deterministic decoding schemes such that $z=f_{\mathcal{D}}(\mathbf{y},\omega$). 

$(ii.)$ A quantum prepare and measure protocol $\mathcal{Q}$ for $(n,\mathcal{R}_i)$ GRAC entails Alice encoding her input onto a qubit $\rho_{\mathbf{x}}$, which is transmitted to Bob. Bob upon receiving the qubit, performs the measurement  $\{M^\mathbf{y}_z | \forall \mathbf{y}: \sum_z M^\mathbf{y}_z =\mathbb{I} \}$ based on his input $\mathbf{y}$ to produce the outcome $z$. The average success probability for such a protocol has the expression  $s^{(n,\mathcal{R}_i)}_{\mathcal{Q}} = \frac{1}{2^n|\mathcal{R}_i|}\sum_{\mathbf{x},\mathbf{y}}\tr(\rho_\mathbf{x}M^{\mathbf{y}}_{z=f_\mathbf{y}(\mathbf{x})})$. The maximal quantum success probability of 
$(n,\mathcal{R}_i)$ GRAC $S^{(n,\mathcal{R}_i)}_{\mathcal{Q}}$ has the expression, 
\small
\begin{equation} \label{prepmeasMax}
 S^{(n,\mathcal{R}_i)}_{\mathcal{Q}}=\max_{\{\rho_{\mathbf{x}}\},\{M^\mathbf{y}_z\}}\left\{\frac{1}{2^n|\mathcal{R}_i|}\sum_{\mathbf{x},\mathbf{y}}\tr\left(\rho_\mathbf{x}M^{\mathbf{y}}_{z=f_\mathbf{y}(\mathbf{x})}\right)\right\},  
\end{equation}
\normalsize
where the maximization is over all two dimensional states $\{\rho_\mathbf{x}\}$ and two dimensional binary outcome measurements $\{M^\mathbf{y}_z\}$. For maximal average success probability it is enough to consider pure states, i.e., $\forall \mathbf{x}:\rho_\mathbf{x} \equiv \ketbra{\psi_\mathbf{x}}{\psi_\mathbf{x}}$, and measurement operators to be projectors, i.e., $\{M^\mathbf{y}_z \equiv \Pi^\mathbf{y}_z|\forall \mathbf{y}: \sum_z\Pi^\mathbf{y}_z=\mathbb{I}\}$. This allows us to re-express \eqref{prepmeasMax} as, 
\small
\begin{equation} \label{prepmeasMax1}
 S^{(n,\mathcal{R}_i)}_{\mathcal{Q}}=\max_{\{\mathbf{r}_{\mathbf{x}}\},\{\mathbf{v}_\mathbf{y}\}}\left\{\frac{1}{2^n|\mathcal{R}_i|}\sum_{\mathbf{x},\mathbf{y}}\frac{1}{2}\left(1+(-1)^{f_{\mathbf{y}}(\mathbf{x})}\mathbf{r}_{\mathbf{x}}.\mathbf{v}_{\mathbf{y}}\right)\right\},  
\end{equation}
\normalsize
where we have used Bloch vector notation for states $\rho_{\mathbf{x}}=\frac{\mathbb{I}+\mathbf{r}_{\mathbf{x}}.\boldsymbol{\sigma}}{2}$, and for measurements $M^{\mathbf{y}}_z=\frac{\mathbb{I}+(-1)^z\mathbf{v}_{\mathbf{y}}.\boldsymbol{\sigma}}{2}$, where $\mathbf{r}_{\mathbf{x}}\in \mathbb{R}_3,\mathbf{v}_{\mathbf{y}}\in \mathbb{R}_3$ are unit vectors, such that $\forall \mathbf{x}: \|\mathbf{r}_{\mathbf{x}}\|=1,\forall \mathbf{y}: \|\mathbf{v}_{\mathbf{y}}\|=1$, and $\boldsymbol{\sigma}$ is the vector of Pauli matrices.

$(iii.)$ An entanglement assisted classical communication protocol ${EACC}$ entails Alice and Bob pre-sharing an entangled quantum state $\rho_{AB}$ of arbitrary local dimension. Alice based on her input measures her part of the entangled state employing the binary outcome measurement $\{M^{\mathbf{x}}_{\omega}|\forall \mathbf{x}: \sum_\omega M^{\mathbf{x}}_{\omega} = \mathbb{I}\}$, and transmits the outcome $\omega$ to Bob. Bob upon receiving the message $\omega$, and his input $\mathbf{y}$ performs the binary outcome measurements $\{M^{\omega,\mathbf{y}}_{z}|\forall \mathbf{y},\omega: \sum_z M^{\omega,\mathbf{y}}_{z} = \mathbb{I}\}$ to produce the outcome $z$. The average success probability for such a protocol has the expression  $s_{{EACC}} = \frac{1}{2^n|\mathcal{R}_i|}\sum_{\mathbf{x},\mathbf{y},\omega}\tr\left(\rho_{AB}M^{\mathbf{x}}_{\omega}\otimes M^{\omega,\mathbf{y}}_{z=f_\mathbf{y}(\mathbf{x})}\right)$.  
The maximal success probability of EACC protocols in $(n,\mathcal{R}_i)$ GRAC, $S_{EACC}$, has the expression, 
\small
\begin{align}
& S_{EACC}=\nonumber\\
& \max_{\rho_{AB},\{M^\mathbf{x}_\omega\},\{M^{\omega,\mathbf{y}}_{z}\}}\left\{\frac{1}{2^n|\mathcal{R}_i|}\sum_{\mathbf{x},\mathbf{y},\omega}\tr\left(\rho_{AB}M^{\mathbf{x}}_{\omega}\otimes M^{\omega,\mathbf{y}}_{z=f_\mathbf{y}(\mathbf{x})}\right)\right\}.  
\end{align}
\normalsize
\subsection{Bounding success of GRACs}
Now we are prepared to present our results for bounding the average success probability of $(n,\mathcal{R}_i)$ GRACs.

\begin{theorem} \label{GRACLB}
The maximal success probability of $(n,\mathcal{R}_i)$ GRACs, $S^{(n,\mathcal{R}_i)}_{\mathcal{O}}$, is lower bounded by that of $(|\mathcal{R}_i|\rightarrow 1)$ RAC, $S^{(|\mathcal{R}_i|\rightarrow 1)}_{\mathcal{O}}$, i.e.,
\begin{equation}
    S^{(n,\mathcal{R}_i)}_{\mathcal{O}} \geq S^{(|\mathcal{R}_i|\rightarrow1)}_{\mathcal{O}},
\end{equation}
where $\mathcal{O}\in\{\mathcal{C},\mathcal{Q},EACC\}$.
\end{theorem}
\begin{proof} To prove the desired thesis we provide a viable strategy for $(n,\mathcal{R}_i)$ GRAC which utilizes an optimal $(|\mathcal{R}_i|\rightarrow1)$ RAC as a subroutine, and achieves success $s^{n,\mathcal{R}_i}_{\mathcal{O}}\geq S^{(|\mathcal{R}_i|\rightarrow1)}_{\mathcal{O}}$.

Given a $(n,\mathcal{R}_i)$ GRAC with the input string $\mathbf{x}\equiv\{x_1,\ldots,x_n\}\in\{0,1\}^n$, consider the bit string $(f_{\mathbf{r}}(\mathbf{x}))_{\mathbf{r}\in \mathcal{R}_i}$. Notice that, $(f_{\mathbf{r}}(\mathbf{x}))_{\mathbf{r}\in \mathcal{R}_i}$ may not be uniformly distributed.
Now, consider a $(|\mathcal{R}_i|\rightarrow1)$ RAC with the input string $\mathbf{\tilde{x}} \equiv \{\tilde{x}_1,\ldots,\tilde{x}_{|\mathcal{R}_i|}\}\in \{0,1\}^{|\mathcal{R}_i|}$ with maximal success probability $S^{|\mathcal{R}_i|}_{\mathcal{O}}=\frac{1}{2^{|R_i|}|R_i|}\sum_{\mathbf{x},\tilde{y}\in\{1,\ldots,|\mathcal{R}_i|\}}p(\tilde{z}=\tilde{x}_{\tilde{y}}|\mathbf{\tilde{x}},\tilde{y})$, where $\mathcal{O}\in\{\mathcal{C},\mathcal{Q},EACC\}$ specifies the particular type of the protocol. We use the string $(f_{\mathbf{r}}(\mathbf{x}))_{\mathbf{r}\in \mathcal{R}_i}$ as the input string for the $(|\mathcal{R}_i|\rightarrow1)$ RAC, up to optimal reordering, i.e., $\mathbf{\tilde{x}}=Perm((f_{\mathbf{r}}(\mathbf{x}))_{\mathbf{r}\in \mathcal{R}_i})$. It is easy to see that, this protocol achieves success probability $s^{(n,\mathcal{R}_i)}_{\mathcal{O}}\geq S^{(|\mathcal{R}_i|\rightarrow1)}_{\mathcal{O}}$, where the inequality is saturated when the optimal strategy of $(|\mathcal{R}_i|\rightarrow1)$ RAC has equal success for all inputs, $\forall \tilde{\mathbf{x}}\in\{0,1\}^{|\mathcal{R}_i|}: \frac{1}{2^{|R_i|}|R_i|}\sum_{\tilde{y}\in\{1,\ldots,|\mathcal{R}_i|\}}p(\tilde{z}=\tilde{x}_{\tilde{y}}|\mathbf{\tilde{x}},\tilde{y})= S^{(|\mathcal{R}_i|\rightarrow1)}_{\mathcal{O}}$.     
\end{proof}
\begin{theorem} \label{UBPrepAndMeas}
The maximal success probability of a prepare and measure protocol in an $(n,\mathcal{R}_i)$ GRAC, $S^{(n,\mathcal{R}_i)}_{\mathcal{Q}}$, is upper bounded as follows.
\begin{equation}
    S^{(n,\mathcal{R}_i)}_{\mathcal{Q}} \leq \frac{1}{2}\left(1+\frac{1}{\sqrt{|\mathcal{R}_i|}}\right),
\end{equation}
\end{theorem}
\begin{proof}
We start by recalling that the maximal success probability of prepare and measure protocol for an $(n,\mathcal{R}_i)$ GRAC has the expression \eqref{prepmeasMax1},
\small
 \begin{eqnarray}\label{plug2}
S^{(n,\mathcal{R}_i)}_{\mathcal{Q}}&=&\max_{\{\mathbf{r}_{\mathbf{x}}\},\{\mathbf{v}_\mathbf{y}\}}\left\{\frac{1}{2^n|\mathcal{R}_i|}\sum_{\mathbf{x},\mathbf{y}}\frac{1}{2}\left(1+(-1)^{f_{\mathbf{y}}(\mathbf{x})}\mathbf{r}_{\mathbf{x}}.\mathbf{v}_{\mathbf{y}}\right)\right\},\nonumber \\ \nonumber
&=&\frac{1}{2}\left(1+\frac{1}{2^n|\mathcal{R}_i|}~\max_{\{\mathbf{r}_{\mathbf{x}}\},\{\mathbf{v}_\mathbf{y}\}}\underbrace{\left\{\sum_{\mathbf{x},\mathbf{y}}(-1)^{f_{\mathbf{y}}(\mathbf{x})}\mathbf{r}_{\mathbf{x}}.\mathbf{v}_{\mathbf{y}}\right\}}_{\Phi_{n,|\mathcal{R}_i|}(\{\mathbf{r}_\mathbf{x}\},\{\mathbf{v}_{\mathbf{y}}\})}\right). \\
\end{eqnarray}
\normalsize
Consequently, finding maximal success probability of prepare and measure protocols in $(n,\mathcal{R}_i)$ GRACs effectively boils down to solving the following optimization problem,
\small
\begin{eqnarray} \label{midProofEq}
\Phi(n,|\mathcal{R}_i|)&=&\max_{\{\mathbf{r}_{\mathbf{x}},\{\mathbf{v}_\mathbf{y}\}}\left\{\Phi_{n,|\mathcal{R}_i|}\left(\{\mathbf{r}_\mathbf{x}\},\{\mathbf{v}_{\mathbf{y}}\}\right)\right\},\nonumber\\
&=&\max_{\{\mathbf{v}_\mathbf{y}\}}\left\{\sum_{\mathbf{x}}\max_{\{\mathbf{r}_{\mathbf{x}}\}}\left\{\mathbf{r}_{\mathbf{x}}.\sum_{\mathbf{y}}(-1)^{f_{\mathbf{y}}(\mathbf{x})}\mathbf{v}_{\mathbf{y}}\right\}\right\},
\end{eqnarray}
\normalsize
Defining, $\mathbf{V}_{\mathbf{x}} = \sum_{\mathbf{y}}(-1)^{f_{\mathbf{y}}(\mathbf{x})}\mathbf{v}_{\mathbf{y}}$, we notice that the scalar product $\mathbf{r}_{\mathbf{x}}.\mathbf{V}_{\mathbf{x}}$ is maximized when $\mathbf{r}_{\mathbf{x}}$ has the same direction as $\mathbf{V}_{\mathbf{x}}$, i.e., when $\mathbf{r}_{\mathbf{x}}=\mathbf{V}_{\mathbf{x}} /\left\|\mathbf{V}_{\mathbf{x}}\right\|$, which implies $\mathbf{r}_{\mathbf{x}}.\mathbf{V}_{\mathbf{x}}=\left\|\mathbf{V}_{\mathbf{x}}\right\|$. This observation allows us to re-express \eqref{midProofEq} as,
\begin{equation}
\Phi(n,|\mathcal{R}_i|)=\max_{\{\mathbf{v}_{\mathbf{y}}\}}\left\{\sum_{\mathbf{x}}\left\|\sum_{\mathbf{y}}(-1)^{f_{\mathbf{y}}(\mathbf{x})}\mathbf{v}_{\mathbf{y}}\right\|\right\}.
\label{eq:max}
\end{equation}
We now further rewrite the equation \eqref{eq:max} as,  
\begin{equation} \label{plug1}
\Phi(n,|\mathcal{R}_i|)=\max_{\{\mathbf{v}_{\mathbf{y}}\}}\underbrace{\left\{\sum_{\mathbf{x}}\left\|\sum_{\mathbf{y}}{g}_{\mathbf{y}}(\mathbf{x})\mathbf{v}_{\mathbf{y}}\right\|\right\}}_{\Phi_{n,|\mathcal{R}_i|}\left(\{\mathbf{v}_{\mathbf{y}}\}\right)};
\end{equation}
where ${g}_{\mathbf{y}}(\mathbf{x})=(-1)^{{f}_{\mathbf{y}}(\mathbf{x})}$. Now, $\Phi_{n,|\mathcal{R}_i|}(\{\mathbf{v}_{\mathbf{y}}\})$ can be thought as the dot product between $\mathbf{z}=(1,1,\cdots,1)\in \mathbb{R}^{2^n}$ and $\mathbf{w}=\left(\left\|\sum_{\mathbf{y}}{g_{\mathbf{y}}(\mathbf{x}_1)}\mathbf{v}_{\mathbf{y}}\right\|,\ldots,\left\|\sum_{\mathbf{y}}{g_{\mathbf{y}}(\mathbf{x}_{2^n})}\mathbf{v}_{\mathbf{y}}\right\| \right)\in \mathbb{R}^{2^n}$, where $\mathbf{x}_1\equiv (x_0=0,\ldots,x_{n}=0), \ldots ,\mathbf{x}_{2^n}\equiv (x_0=1,\ldots,x_{n}=1)$. Now recall that Cauchy-Schwarz inequality implies,
\begin{equation}
\Phi_{n,|\mathcal{R}_i|}\left(\{\mathbf{v}_{\mathbf{y}}\}\right)=\boldsymbol{z} \cdot \boldsymbol{w} \leq\|\boldsymbol{z}\| \| \boldsymbol{w}\|.
\end{equation}
Now observe that,
\begin{eqnarray} \label{TheSum}
\| \boldsymbol{w}\|^2&=&\sum_{\mathbf{x}}\left\|\sum_{\mathbf{y}}{g_{\mathbf{y}}(\mathbf{x})}\mathbf{v}_{\mathbf{y}}\right\|^2\nonumber\\
&=& \sum_{\mathbf{x}}\left(\sum_{\mathbf{y}}{g_{\mathbf{y}}(\mathbf{x})}\mathbf{v}_{\mathbf{y}}\cdot\sum_{\mathbf{y}'}{g_{\mathbf{y}'}(\mathbf{x})}\mathbf{v}_{\mathbf{y}'}\right). 
\end{eqnarray}
There are two types of terms that appear in the sum in \eqref{TheSum}, $(i.)$ 
whenever $\mathbf{y}= \mathbf{y}'$, in this case $\forall \mathbf{y}\in \mathcal{R}_i: {g_{\mathbf{y}}(\mathbf{x})}\mathbf{v}_{\mathbf{y}}\cdot{g_{\mathbf{y}'}(\mathbf{x})} = \|\mathbf{v}_{\mathbf{y}}\|^2$, $(ii.)$ whenever $\mathbf{y} \neq \mathbf{y}'$, we have terms of the form ${g_{\mathbf{y}}(\mathbf{x})} {g_{\mathbf{y}'}(\mathbf{x})}\mathbf{v}_{\mathbf{y}}\cdot\mathbf{v}_{\mathbf{y}'}$. Now as all functions in our MUBS are balanced and pairwise mutually unbiased, there exits $2^{n-1}$ strings $\mathbf{x} \in (\mathcal{X}_{f_{\mathbf{y}}=0} \cap \mathcal{X}_{f_{\mathbf{y}'}=0})\cup(\mathcal{X}_{f_{\mathbf{y}}=1} \cap \mathcal{X}_{f_{\mathbf{y}'}=1})$ for which the coefficients  ${g_{\mathbf{y}}(\mathbf{x})} {g_{\mathbf{y}'}(\mathbf{x})}=1$, and there exists $2^{n-1}$ strings $\mathbf{x} \in  (\mathcal{X}_{f_{\mathbf{y}}=0} \cap \mathcal{X}_{f_{\mathbf{y}'}=1})\cup(\mathcal{X}_{f_{\mathbf{y}}=1} \cap \mathcal{X}_{f_{\mathbf{y}'}=1})$ for which the coefficients  ${g_{\mathbf{y}}(\mathbf{x})} {g_{\mathbf{y}'}(\mathbf{x})}=-1$. Consequently, all terms of the form $(ii.)$ cancel out, and we are left with,
\begin{eqnarray}
\| \boldsymbol{w}\|^2&=&\sum_{\mathbf{x}\in\{0,1\}^n,\mathbf{y}\in\mathcal{R}_i}\|\mathbf{v}_{\mathbf{y}}\|^2=2^n|\mathcal{R}_i|,\nonumber\\
\implies~\| \boldsymbol{w}\|&=&\sqrt{2^n|\mathcal{R}_i|}.
\end{eqnarray}
Since $\| \boldsymbol{z}\|=\sqrt{2^n}$, we have $\Phi_{n,|\mathcal{R}_i|}(\{\mathbf{v}_{\mathbf{y}}\})\le 2^n\sqrt{|\mathcal{R}_i|}$, which when plugged back in \eqref{plug1} and \eqref{plug2} yields the desired thesis,
\begin{eqnarray}
S^{(n,\mathcal{R}_i)}_{\mathcal{Q}}\le\frac{1}{2}\left(1+\frac{2^n\sqrt{|\mathcal{R}_i|}}{2^n|\mathcal{R}_i|}\right)=\frac{1}{2}\left(1+\frac{1}{\sqrt{\mathcal{R}_i}}\right).
\end{eqnarray}
\end{proof}
\begin{theorem} \label{EACCBineq}
The maximal success probability of an EACC protocol in an $(n,\mathcal{R}_i)$ GRAC, $S^{(n,\mathcal{R}_i)}_{EACC}$, is upper bounded by the maximum quantum value of the following Bell expression $B^{(n,\mathcal{R}_i)}_\mathcal{Q}$, i.e, $S^{(n,\mathcal{R}_i)}_{EACC}\leq B^{(n,\mathcal{R}_i)}_\mathcal{Q}$
\small
\begin{eqnarray} \label{BellEacc} \nonumber
    B^{(n,\mathcal{R}_i)} \equiv \frac{1}{2^n|\mathcal{R}_i|} \sum_{\mathbf{x},\mathbf{y},y_0\in\{0,1\}}p(u=y_0,v=f_{\mathbf{y}}(\mathbf{x}) |\mathbf{x},y_0,\mathbf{y}), \\
\end{eqnarray}
\normalsize
where $\mathbf{x} \in \{0,1\}^n$ is Alice's input, $(y_0\in\{0,1\},\mathbf{y}\in|\mathcal{R}_i|)$ are Bob's input, and $u,v\in\{0,1\}$ are outputs of Alice and Bob, respectively.
\end{theorem}
\begin{proof}
 To prove the above thesis all we need to demonstrate is that for an EACC $(n,\mathcal{R}_i)$ GRAC achieving success probability $S^{(n,\mathcal{R}_i)}_{EACC}=\frac{1}{2^n|\mathcal{R}_i|}\sum_{\mathbf{x},\mathbf{y},\omega}\tr(\rho_{AB}M^{\mathbf{x}}_{\omega}\otimes M^{\omega,\mathbf{y}}_{z=f_\mathbf{y}(\mathbf{x})})$, a quantum correlation can be obtained that achieves the same value for the Bell expression in \eqref{BellEacc}, so that the maximum value of the Bell expression caps the success probability of the EACC protocol. To this end we recall that an EACC protocol for $(n,\mathcal{R}_i)$ GRAC entails a pre-shared entangled state $\rho_{AB}$, local measurements for Alice $\{M^{\mathbf{x}}_{\omega}\}$, and local measurements for Bob $\{M^{\omega,\mathbf{y}}_{z}\}$, where Alice's output $\omega$ is transmitted to Bob. Now instead of transmitting the output of her measurement, Alice simply relabels it as her local output, i.e., $u=\omega$. On the other hand, Bob obtains an additional uniformly sampled input bit $y_0\in\{0,1\}$, utilizing it instead of the message from Alice to decide on the measurement $\{M^{y_0,\mathbf{y}}_{z}\}$ he performs on his part of the entangled state. Finally, Bob relabels his output as $v=z$. As a result, they obtain the correlation $p(u,v|\mathbf{x},y_0,\mathbf{y})=\tr(\rho_{AB}M^{\mathbf{x}}_{u}\otimes M^{y_0,\mathbf{y}}_{v}\})$. Clearly, this correlation because of the construction achieves the Bell value
 \small
\begin{align*}
B^{(n,\mathcal{R}_i)}_{Q}&=\frac{1}{2^n|\mathcal{R}_i|}\sum_{\mathbf{x},\mathbf{y},y_0\in\{0,1\}}\tr\left(\rho_{AB}M^{\mathbf{x}}_{u=y_0}\otimes M^{y_0,\mathbf{y}}_{v=f_{\mathbf{y}}(\mathbf{x})}\right)\\
&=\frac{1}{2^n|\mathcal{R}_i|}\sum_{\mathbf{x},\mathbf{y},\omega}\tr\left(\rho_{AB}M^{\mathbf{x}}_{\omega}\otimes M^{\omega,\mathbf{y}}_{z=f_\mathbf{y}(\mathbf{x})}\right) = S^{(n,\mathcal{R}_i)}_{EACC}   
\end{align*}
\normalsize
Therefore, the maximum Bell value of the Bell expression \eqref{BellEacc}, $B^{(n,\mathcal{R}_i)}_{Q}$ caps the success probability of EACC protocols in $(n,\mathcal{R}_i)$ GRAC, $S^{(n,\mathcal{R}_i)}_{EACC}$. 
\end{proof}

\section{$\mathbf{n=3}$: a case study}\label{sec4}
In this section, we study and characterize $(n=3,\mathcal{R}_i)$ GRACs, finding out optimal classical and quantum protocol and success probabilities, as well as noise tolerance of the latter. 
\subsection{Classical protocols}
We recall that in classical $(n=3,\mathcal{R}_i)$ GRACs, Alice encodes her input string $\mathbf{x}\in \{0,1\}^3$ onto a classical bit message $\omega \in \{0,1\}$, based on a deterministic encoding scheme $\mathcal{E}$, $\omega=f_{\mathcal{E}}(\mathbf{x})$. On the other end, Bob upon receiving $\omega$ from Alice, decodes it to produce his output $z=f_{\mathcal{D}}(\mathbf{y},\omega)$ based on a deterministic decoding scheme $\mathcal{D}$. The optimal classical strategy for $(n\rightarrow 1)$ RACs, without loss of generality, turns out to be \textit{majority encoding}, i.e., $\omega=maj(x_1,\ldots,x_n)$, and \textit{identity decoding}, i.e., $z=\omega$ \cite{Ambainis08,Ambainis15}. However, as we demonstrate below, this strategy may not be optimal for $(n=3,\mathcal{R}_i)$ GRACs. 
\onecolumngrid
\begin{center}
	\begin{table}[h!]
		\begin{tabular}{@{}c|c|c|c|c|c|c|c|c|c|@{}} \toprule
			$\mathbf{x}$ & \textcolor[rgb]{0,.5,0}{$\omega=maj(x_1,x_2,x_3)$} & \textcolor{blue}{$\omega = x_1\wedge \neg(x_2\wedge x_3)$ } & $x_1$                             & $x_2$                                         & $x_3$                                         & $x_1\oplus x_2$                   & $x_1\oplus x_3$                   & $x_2\oplus x_3$                   & $x_1\oplus x_2 \oplus x_3$                    \\ \toprule
			$(000)$      & \textcolor[rgb]{0,.5,0}{$0$}                          & \textcolor{blue}{$0$}                      & $0$                               & $0$                                           & $0$                                           & $0$                               & $0$                               & $0$                               & $0$                                           \\
			$(001)$      & \textcolor[rgb]{0,.5,0}{$0$}                          & \textcolor{blue}{$0$}                      & $0$                               & $0$                                           & $1$                                           & $0$                               & $1$                               & $1$                               & $1$                                           \\
			$(010)$      & \textcolor[rgb]{0,.5,0}{$0$}                          & \textcolor{blue}{$0$}                      & $0$                               & $1$                                           & $0$                                           & $1$                               & $0$                               & $1$                               & $1$                                           \\
			$(011)$      & \textcolor[rgb]{0,.5,0}{$1$}                          & \textcolor{blue}{$0$}                      & $0$                               & $1$                                           & $1$                                           & $1$                               & $1$                               & $0$                               & $0$                                           \\
			$(100)$      & \textcolor[rgb]{0,.5,0}{$0$}                          & \textcolor{blue}{$1$}                      & $1$                               & $0$                                           & $0$                                           & $1$                               & $1$                               & $0$                               & $1$                                           \\
			$(101)$      & \textcolor[rgb]{0,.5,0}{$1$}                          & \textcolor{blue}{$1$}                      & $1$                               & $0$                                           & $1$                                           & $1$                               & $0$                               & $1$                               & $0$                                           \\
			$(110)$      & \textcolor[rgb]{0,.5,0}{$1$}                          & \textcolor{blue}{$1$}                      & $1$                               & $1$                                           & $0$                                           & $0$                               & $1$                               & $1$                               & $0$                                           \\
			$(111)$      & \textcolor[rgb]{0,.5,0}{$1$}                          & \textcolor{blue}{$0$}                      & $1$                               & $1$                                           & $1$                                           & $0$                               & $0$                               & $0$                               & $1$                                           \\ \midrule \midrule
			\multicolumn{3}{c|}{\multirow{2}{*}{$s^{(n=3,\mathcal{R})}_{\mathcal{C}}(\mathbf{y})$}}                          & \textcolor[rgb]{0,.5,0}{$6/8$}    & \textcolor[rgb]{0,.5,0}{$6/8$}                & \textcolor[rgb]{0,.5,0}{$6/8$}                & \textcolor[rgb]{0,.5,0}{$4/8$}    & \textcolor[rgb]{0,.5,0}{$4/8$}    & \textcolor[rgb]{0,.5,0}{$4/8$}    & \textcolor[rgb]{0,.5,0}{$6/8^{(\star)}$}      \\
			\multicolumn{3}{c|}{}                                                                                            & \textcolor{blue}{$7/8[\uparrow]$} & \textcolor{blue}{$5/8^{(\star)}[\downarrow]$} & \textcolor{blue}{$5/8^{(\star)}[\downarrow]$} & \textcolor{blue}{$5/8[\uparrow]$} & \textcolor{blue}{$5/8[\uparrow]$} & \textcolor{blue}{$5/8[\uparrow]$} & \textcolor{blue}{$5/8^{(\star)}[\downarrow]$} \\ \bottomrule
		\end{tabular}
		\caption{(Color online) \label{MajIsNotOptimal}
		Explicit comparison of classical strategies, $(i.)$ \textit{majority-encoding} and \textit{identity-decoding} (green), and $(ii.)$ An encoding strategy entailing $\omega = x_1 \wedge \neg(x_2 \wedge x_3)$, along with the decoding scheme entailing producing $z=\omega \oplus 1$ whenever Bob is asked to guess $x_2$, $x_3$, or $x_1\oplus x_2 \oplus x_3$, and $z=\omega$ otherwise, for $(n=3,\mathcal{R})$ GRAC. We also enlist the guess probability for the functions $f_i\in\mathcal{F}^3_{\mathcal{R}}$ in the bottom row, for the ease of access. The asterisk $(\star)$ indicates the use of \textit{inverse identity} decoding, i.e., $z=\omega \oplus 1$ for the particular function.}
	\end{table}
\end{center}
\twocolumngrid

\begin{observation}
Unlike $(n\rightarrow1)$ RAC, \textit{majority encoding} and \textit{indentity decoding} is not optimal for all $(n=3,\mathcal{R}_i)$ GRACs,
\end{observation}
\begin{proof}
To prove this thesis, we shall consider the $(n=3,\mathcal{R})$ GRAC which entails the entire MUBS $\mathcal{F}^3_{\mathcal{R}}$. For this task, the \textit{majority encoding}, i.e., $\omega=maj(x_1,x_2,x_3)$, and \textit{identity decoding}, i.e., $z=\omega$ strategy yields a success probability of $s^{(n,\mathcal{R})}_{\mathcal{C}}=\frac{32}{56}$. Moreover, even if we allow for a \textit{inverse identity decoding}, for the function $x_1\oplus x_2\oplus x_3$, i.e., $z=\omega\oplus 1$, we obtain success probability of $s^{(3,\mathcal{R})}_{\mathcal{C}}=\frac{36}{56}$ (see Table \ref{MajIsNotOptimal}). However, employing a straightforward linear program, we find that the optimal success probability of $(n=3,\mathcal{R})$ GRAC turns out to be $S^{(3,\mathcal{R})}_{\mathcal{C}}=\frac{37}{56}$. Specifically, Alice's encoding strategy entails $\omega = x_1 \wedge \neg(x_2 \wedge x_3)$, whereas Bob's decoding scheme entails producing $z=\omega \oplus 1$ whenever he is asked to guess $x_2$, $x_3$, or $x_1\oplus x_2 \oplus x_3$, and $z=\omega$ otherwise. We note that this strategy is not unique, as there exist other strategies that saturate the classical optimal success probability.

\end{proof}
 Moving on, we used straightforward linear programs to obtain the optimal classical success probabilities for all non-trivial $(n=3,\mathcal{R}_i)$ GRACs (see Table \ref{classical}). 
\begin{table}[b]
\begin{tabular}{@{}c|c|c@{}}
	\toprule
		$|\mathcal{R}_i|$    & $S^{(|\mathcal{R}_i|\rightarrow 1)}_{\mathcal{C}}$ & $S^{(3,\mathcal{R}_i)}_{\mathcal{C}}$                                                                        \\ \toprule
		$2$                  & $\frac{3}{4} = 0.75$                         & $\frac{3}{4} = 0.75$                                                                                                \\ \midrule
		$3$                  & $\frac{3}{4} = 0.75$                         & $\frac{3}{4} = 0.75$                                                                                                \\ \midrule
		\multirow{2}{*}{$4$} & \multirow{2}{*}{$\frac{11}{16} = 0.6875$}      & \begin{tabular}[c]{@{}c@{}}$\frac{3}{4} = 0.75$ \\ $\ (\text{if } f_i\oplus f_j = f_k \oplus f_l)$\end{tabular}     \\ \cmidrule(l){3-3} 
		&                                       & \begin{tabular}[c]{@{}c@{}}$\frac{11}{16} = 0.6875$\\  $\ (\text{if} f_i\oplus f_j \neq f_k \oplus f_l)$\end{tabular} \\ \midrule
		$5$                  & $\frac{11}{16} = 0.6875$                       & $\frac{7}{10} = 0.7$                                                                                               \\ \midrule
		$6$                  & $\frac{21}{32} = 0.65625$                       & $\frac{2}{3} \approx 0.67$                                                                                                \\ \midrule
		$7$                  & $\frac{21}{32}= 0.65625$                       & $\frac{37}{56}\approx 0.66$                                                                                              \\ \bottomrule
\end{tabular}
\caption{\label{classical} The maximal classical success probability of $(3,\mathcal{R}_i)$ GRACs, $S^{(3,\mathcal{R}_i)}_{\mathcal{C}}$, listed along with the number of MUBFs Bob's required to guess. These values are contrasted against the maximal success probability of standard $(n=|\mathcal{R}_i|\rightarrow 1)$ RACs, $S^{(|\mathcal{R}_i|\rightarrow 1)}_{\mathcal{C}}$, which form lower bounds for the former according to Theorem \ref{GRACLB}. These were obtained via linear programming, and by retrieving explicit classical strategies. The case of four MUBFs, presents a peculiarity, i.e., when the four functions $\{f_i,f_j,f_k,f_l\}$ as such that $f_i\oplus f_j = f_k \oplus f_l$ the classical protocols can attain a success probability of $0.75$, whereas in the other cases, classical protocols cannot go beyond $\frac{11}{16}$, which is also the maximal success probability of $(4\rightarrow 1)$ RAC (see Observation \ref{successDependsOnQuestions}). Moreover, notice that for the latter case, adding any MUBF to the latter increases the maximal classical average success probability, demonstrating the surprising feature of GRACs termed \textit{harder the task, greater the payoff} (see Observation \ref{hard}).}
\end{table}
 We find that the optimal strategies for all cases, are (equivalent up to rebelling) to either \textit{majority encoding} and \textit{identity decoding}, or the strategy entailing the encoding $\omega = x_1 \wedge \neg(x_2 \wedge x_3)$. While for $|\mathcal{R}_i|\in\{2,3\}$, we find the maximal classical success probability of the $(n=3,\mathcal{R}_i)$ GRACs remains the same as that of the $(|\mathcal{R}_i|\rightarrow 1)$ RACs. For $|\mathcal{R}_i|=\{5,6\}$, the maximal classical success probability of the $(n=3,\mathcal{R}_i)$ GRACs exceeds that of the $(|\mathcal{R}_i|\rightarrow 1)$ RACs. The case of four questions, $|\mathcal{R}_i|=4$, presents a peculiarity, and forms the basis of the following observation.
 \begin{observation} \label{successDependsOnQuestions}
 The maximum average success probability of $(n=3,\mathcal{R}_i)$ GRAC depends on the list of functions that Bob needs to guess, when $|\mathcal{R}_i|=4$.
 \end{observation}
 \begin{proof}
 Consider the case, when Bob is required to guess $\mathcal{F}^3_{\mathcal{R}_i}\equiv\{x_1,x_2,x_3,x_1\oplus x_2 \oplus x_3\} \in \mathcal{F}^3_{\mathcal{R}}$, then \textit{majority encoding} and \textit{identity decoding} yields the optimal success probability, $S^{(n=3,\mathcal{R}_i)}_{\mathcal{C}} = \frac{3}{4}$. In general, whenever Bob has to guess $\mathcal{F}^3_{\mathcal{R}_i}\equiv\{f_i,f_j,f_k,f_l \in \mathcal{F}^3_{\mathcal{R}}\}$ such that $\forall \mathbf{x} \in\{0,1\}^3: f_i(\mathbf{x}) \oplus f_j(\mathbf{x}) = f_k(\mathbf{x})\oplus f_l(\mathbf{x})$, a strategy equivalent upto relabelling \textit{majority encoding} and \textit{identity decoding} yields the optimal success probability, $S^{(n=3,\mathcal{R}_i)}_{\mathcal{C}} = \frac{3}{4}$. 
 
 Now consider the case when Bob is required to guess $\mathcal{F}^3_{\mathcal{R}_i}\equiv\{x_1,x_2,x_3,x_1\oplus x_2\} \in \mathcal{F}^3_{\mathcal{R}}$. In this case the strategy entailing the encoding scheme $\omega=x_1 \wedge \neg(x_2 \wedge x_3)$, and the decoding scheme $z=\omega$ when Bob is asked to guess $x_1$ or $x_1\oplus x_2$, and $z = \omega \oplus 1$ otherwise, attains the optimal success probability $S^{(n=3,\mathcal{R}_i)}_{\mathcal{C}} = \frac{11}{16}$. In fact, whenever, $\exists \mathbf{x} \in\{0,1\}^3: f_i(\mathbf{x}) \oplus f_j(\mathbf{x}) \neq f_k(\mathbf{x})\oplus f_l(\mathbf{x})$, a strategy equivalent upto relabelling to the aforementioned strategy attains the optimal success probability, $S^{(n=3,\mathcal{R}_i)}_{\mathcal{C}} = \frac{11}{16}$. 
 \end{proof}
 The case of $(n=3,\mathcal{R}_i)$ GRAC with four questions, i.e., when $|\mathcal{R}_i|=4$, presents yet another peculiarity.  
 \begin{observation} \label{hard}[Harder the task, greater the payoff]
 In the case of four questions, $|\mathcal{R}_i|=4$, and $\exists \mathbf{x} \in\{0,1\}^3: f_i(\mathbf{x}) \oplus f_j(\mathbf{x}) \neq f_k(\mathbf{x})\oplus f_l(\mathbf{x})$, the maximal average success probability increases from $\frac{11}{16}=0.6875$ to $\frac{7}{10}=0.7$ when Bob is asked to guess any additional mutually unbiased balanced function of Alice's input. 
 \end{observation} 
 This is especially surprising as, in general, for $(n \rightarrow 1)$ RAC and generic communication complexity tasks, increasing the number of questions that Bob is required to guess, $n$, decreases the maximal average success probability (see Table \ref{classical}).
\begin{table}[]
\begin{tabular}{@{}c|c|c@{}}
		\toprule
		$|\mathcal{R}_i|$    & $S^{(|\mathcal{R}_i|\rightarrow 1)}_{\mathcal{Q}}$ & $S^{(3,\mathcal{R}_i)}_{\mathcal{Q}}$                                                                        \\ \toprule
		$2$                  & $\frac{1}{2}(1+\frac{1}{\sqrt{2}}) \approx 0.85$                         & $\frac{1}{2}(1+\frac{1}{\sqrt{2}}) \approx 0.85$                                                                                                \\ \midrule
		$3$                  & $\frac{1}{2}(1+\frac{1}{\sqrt{3}}) \approx 0.78$                        & $\frac{1}{2}(1+\frac{1}{\sqrt{3}}) \approx 0.78$                                                                                                \\ \midrule
		
		\multirow{2}{*}{$4$} & \multirow{2}{*}{$\frac{1}{2}(1+\frac{\sqrt{2}+\sqrt{6}}{8})\approx 0.74$}      & \begin{tabular}[c]{@{}c@{}}$\frac{3}{4}=0.75$ \\ $\ (\text{if } f_i\oplus f_j = f_k \oplus f_l)$\end{tabular}     \\ \cmidrule(l){3-3} 
		&                                       & \begin{tabular}[c]{@{}c@{}}$\frac{1}{2}(1+\frac{\sqrt{2}+\sqrt{6}}{8})\approx 0.74$\\  $\ (\text{if} f_i\oplus f_j \neq f_k \oplus f_l)$\end{tabular} \\ \midrule
		$5$                  & $\approx 0.7135$                       & $\frac{1}{2}(1+\frac{1}{\sqrt{5}}) \approx 0.72$                                                                                               \\ \midrule
		$6$                  & $\approx 0.6940$                       & $\frac{1}{2}(1+\frac{1}{\sqrt{6}}) \approx 0.70$                                                                                                \\ \midrule
		$7$                  & $\approx 0.6786$                       & $\frac{1}{2}(1+\frac{1}{\sqrt{7}}) \approx 0.69$                                                                                              \\ \bottomrule
\end{tabular}
\caption{\label{quantum} The maximal quantum success probability of prepare and measure qubit protocols for $(3,\mathcal{R}_i)$ GRACs, $S^{(3,\mathcal{R}_i)}_{\mathcal{Q}}$, listed along with the number of MUBFs Bob is required to guess. These values are contrasted against the maximal quantum success probability of standard $(n=|\mathcal{R}_i|\rightarrow 1)$ RACs, $S^{(|\mathcal{R}_i|\rightarrow 1)}_{\mathcal{Q}}$, which form lower bounds for the former according to Theorem \ref{GRACLB}. All values were obtained upon coincidence (upto numerical precision) of lower bounds obtained from see-saw semi-definite programming method, and upper bounds obtained via \textit{Navascues-Vertesi} hierarchy of semidefinite programming relaxations, along with retrieval of explicit quantum protocols. In all cases expect when $|\mathcal{R}_i|=4$, notice that the maximal quantum success probabilities saturate the upper-bounds, $\frac{1}{2}\left(1+\frac{1}{\sqrt{|\mathcal{R}_i}|}\right)$, which follow from Theorem \ref{UBPrepAndMeas}. In particular, for the case of four MUBFs, when the four functions $\{f_i,f_j,f_k,f_l\}$ are such that $f_i\oplus f_j = f_k \oplus f_l$ the qubit prepare and measure protocols cannot exceed the classical maximum success probability, $0.75$, where in general qubit protocols go beyond the classical bound, $\frac{11}{16}=0.6875$, but saturate the maximal success probability of qubits for $(4\rightarrow 1)$ RAC, $\frac{1}{2}\left(1+\frac{\sqrt{2}+\sqrt{6}}{8}\right)\approx 0.74$.}
\end{table}

\subsection{Quantum prepare and measure protocols}
We now investigate the performance of qubit prepare and measure protocols in $(n=3,\mathcal{R}_i)$ GRACs. We employed standard \textit{see-saw} semi-definite programming technique to obtain lower bounds on maximal success probability of such protocols. Additionally, we employed the \textit{Navascues-Vertesi} hierarchy of semidefinite programming relaxations to obtain tight upper bounds. Whenever, the lower and upper bounds coincide (up to machine precision), they yield a proof of optimality. The consequent optimal values are listed in Table \ref{quantum}. Additionally, we retrieve explicit quantum protocols which saturate these values (see Appendix Section \ref{ExplicitAppend}).  
\subsubsection{Noisy Channels}
In this section, we investigate the effect of noisy channels on the performance of qubit prepare and measure protocols in $(n=3,\mathcal{R}_i)$ GRACs. Recall that a quantum channel is mathematically described by a completely positive trace preserving map $\Lambda:L(\mathcal{H}_{in})\to L(\mathcal{H}_{out})$, where $L(\mathcal{X})$ be the set of linear operators acting on the Hilbert space $\mathcal{X}$; $\mathcal{H}_{in}~\&~ \mathcal{H}_{out}$ respectively denote input and output Hilbert space of the map $\Lambda$ \cite{Kraus83,Chuang00,Wilde13}. Since we are considering qubit communication therefore we have $\mathcal{H}_{in})\equiv\mathcal{H}_{out}\equiv\mathbb{C}^2$. Furthermore, a channel is known to be unital if completely mixed state remains invariant under it. In the following we analyze effect of the following two unital qubit channels \cite{Chuang00,Wilde13} on the performance of $(n=3,\mathcal{R}_i)$ GRACs.
\begin{figure}[b!]
	\centering
	\includegraphics[width=0.75\linewidth]{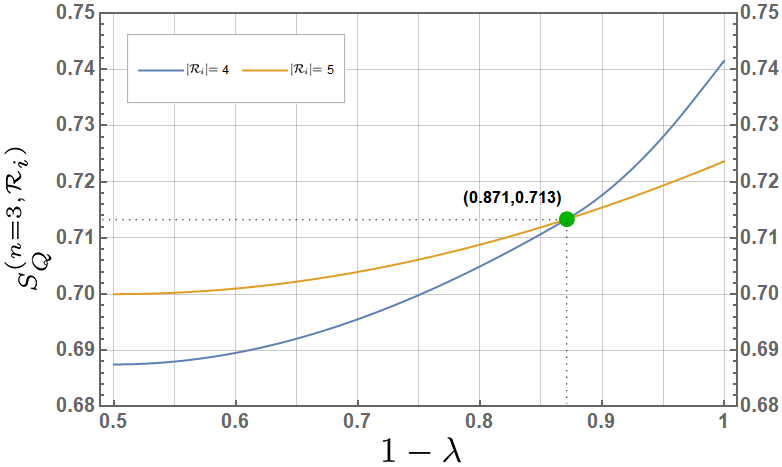}
	\caption{The maximal quantum success probability of $(n=3,\mathcal{R}_i)$ GRAC with $|\mathcal{R}_i|=\{4,5\}$ in presence of dephasing channel with noise parameter $\lambda$. Note that for $|\mathcal{R}_i|=4$ we considered only the cases for which $f_i\oplus f_j \neq f_k\oplus f_l$. Moreover, we find that for a range of the noise parameter $1-\lambda \in (0.5,0.871)$ the maximal quantum success probability $(n=3,\mathcal{R}_i)$ GRAC with four MUBFs exceeds that of $(n=3,\mathcal{R}_i)$ GRAC with five MUBFs, which is yet another instance of \textit{harder the task, greater the payoff}.}
    \label{interPNG}
\end{figure}
\begin{figure}[b!]
	\centering
	\includegraphics[width=0.75\linewidth]{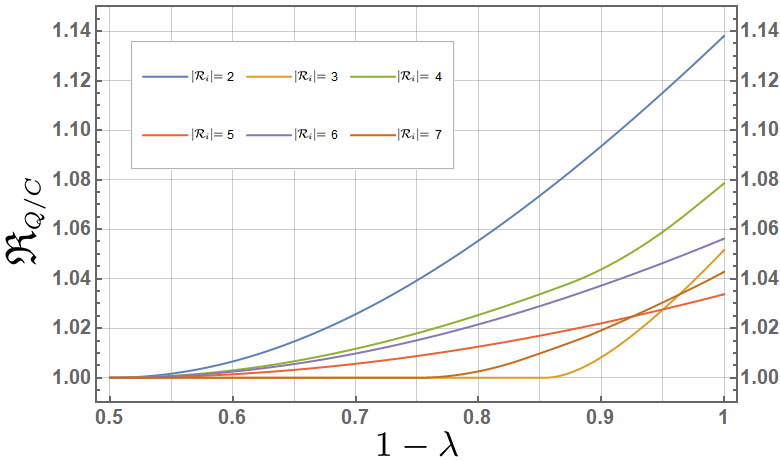}
	\caption{The ratio of maximal quantum success probability to maximal classical success probability, $\mathfrak{R}_{\scaleto{Q/C}{4.5pt}}$ = $\frac{S^{(n=3,\mathcal{R}_i)}_Q}{S^{(n=3,\mathcal{R}_i)}_C}$, for  $(n=3,\mathcal{R}_i)$ GRACs with $|\mathcal{R}_i|$=[2,7] plotted against $1-\lambda$ where $\lambda$ is the the noise parameter of the dephasing channel.}
    \label{ratioPNG}
\end{figure}

$(a.)$ \textit{Depolarising channel:} The effect of a depolarizing channel $\Phi^{\lambda}_{Depol}$ is to keep the input state intact with probability $(1-\lambda)$, while with probability $\lambda$ an 'error' occurs entirely replacing the input state by white noise, i.e., a generic initial state $\rho_{in} = \frac{\mathbb{I}+\mathbf{n}\cdot\boldsymbol{\sigma}}{2}$ is distorted to,
    \begin{equation}
         \rho_{out} = \Phi^{\lambda}_{Depol}(\rho_{in}) = \lambda\frac{\mathbb{I}}{2}+(1-\lambda)\rho_{in},  
    \end{equation}
    where $\lambda \in [0,1]$ is the noise parameter. Now, for qubits, increasing the noise parameter $\lambda$ shrinks the Bloch sphere uniformly, so it is enough to consider the noisy versions of the optimal preparations we recovered above. In the Table \ref{DepolTable} we list the threshold value of the noise parameter $\lambda$ such that we continue to retrieve a quantum advantage in $(n=3,\mathcal{R}_i)$ GRACs. Yet again, we witness the reappearence of the characteristic phenomenon of GRACs, namely, \textit{harder the task, greater the payoff}, as the noise threshold $\lambda_{crit}$ in the cases with $|\mathcal{R}_i|=5,6,7$ MUBFs exceeds that of $|\mathcal{R}_i|=3$. 

\begin{table}[t]
\begin{tabular}{@{}c|c|c|c|c|c|c@{}}
			\toprule
			$|\mathcal{R}_i|$ & $2$       & $3$       & $4$       & $5$       & $6$       & $7$       \\ \midrule
			$\lambda_{crit}$  & $0.29289$ & $0.13396$ & $0.22354$ & $0.10555$ & $0.18349$ & $0.14957$ \\ \bottomrule
\end{tabular}
\caption{\label{DepolTable} Threshold value of the noise parameter $\lambda_{crit}$ for depolarizing channel, such that we continue to retrieve the quantum advantage in $n=3,\mathcal{R}_i$ GRACs. Note that for $|\mathcal{R}_i|=4$ we considered only the cases for which $f_i\oplus f_j \neq f_k\oplus f_l$. Consequently, we observe that for $|\mathcal{R}_i|=3$ the noise tolerance is lower than that of $|\mathcal{R}_i|=5,6,7$, which forms yet another instance of \textit{harder the task, greater the payoff}.} 
\end{table}

$(b.)$ \textit{Dephasing channel:} The effect of a qubit dephasing channel $\Phi^{\lambda}_{Dephase}$ along a given spin-direction $\mathbf{n}$ is given by,
    \begin{equation}
         \rho_{out} = \Phi^{\lambda}_{Dephase}(\rho_{in}) = \lambda (\mathbf{n}\cdot\boldsymbol{\sigma})\rho_{in} (\mathbf{n}\cdot\boldsymbol{\sigma})+(1-\lambda)\rho_{in},  
    \end{equation}
    where $\lambda \in [0,\frac{1}{2}]$ is the noise parameter. Unlike the depolarising channel, here, the optimal quantum strategy, always performs as well as the optimal classical strategy. Therefore, in Fig. \ref{ratioPNG} we plot the ratio optimal quantum success probability to that of maximal classical success probability, $\mathfrak{R}_{\scaleto{Q/C}{4.5pt}}$ = $\frac{S^{(n=3,\mathcal{R}_i)}_Q}{S^{(n=3,\mathcal{R}_i)}_C}$, for different numbers of MUBFs, wherein we employed a \textit{bit-flip} channel, i.e. $\mathbf{n}\equiv [1,0,0]^T$, and numerically optimized over qubit preparations and measurements for all $\lambda \in [0,1]$. Moreover, even in this case we find the reappearance of the phenomenon, \textit{harder the task, greater the payoff}, as for a range of the noise parameter $1-\lambda \in (0.5,0.871)$ the maximal quantum success probability $(n=3,\mathcal{R}_i)$ GRAC with four MUBFs exceeds that of $(n=3,\mathcal{R}_i)$ GRAC with five MUBFs (see Fig. \ref{interPNG}).

\subsection{Entanglement assisted classical communication}
Finally, we investigate the performance of shared entanglement and cbit communication protocols. Again, we employ standard \textit{see-saw} semi-definite programming technique to obtain dimension dependent lower bounds on maximal success probability of such protocols. Moreover, we employ \textit{Navascues-Pironio-Acin} hierarchy of semidefinite programming relaxations to obtain upper bounds on the quantum violation of associated (Theorem \ref{EACCBineq}) Bell inequalities. Yet again, a coincidence (up to machine precision) implies the optimality of these bounds, listed in Table \ref{entanglement}.  
\begin{table}[t]
	\begin{tabular}{@{}c|c|c@{}}
		\toprule
		$|\mathcal{R}_i|$    & $S^{(|\mathcal{R}_i|\rightarrow 1)}_{EACC}$ & $S^{(3,\mathcal{R}_i)}_{EACC}$                                                                        \\ \toprule
		$2$                  & $\frac{1}{2}(1+\frac{1}{\sqrt{2}}) \approx 0.85$                         & $\frac{1}{2}(1+\frac{1}{\sqrt{2}}) \approx 0.85$                                                                                                \\ \midrule
		$3$                  & $\frac{1}{2}(1+\frac{1}{\sqrt{3}}) \approx 0.78$                        & $\frac{1}{2}(1+\frac{1}{\sqrt{3}}) \approx 0.78$                                                                                                \\ \midrule
		$4$                  & $\frac{1}{2}(1+\frac{\sqrt{2}+\sqrt{6}}{8})\approx 0.74$                        & $\frac{1}{2}(1+\frac{1}{\sqrt{4}})=\frac{3}{4}=0.75$                                                                                                \\ \midrule
		
		$5$                  & $\approx 0.7135$                       & $\frac{1}{2}(1+\frac{1}{\sqrt{5}}) \approx 0.72$                                                                                               \\ \midrule
		$6$                  & $\approx 0.6940$                       & $\frac{1}{2}(1+\frac{1}{\sqrt{6}}) \approx 0.70$                                                                                                \\ \midrule
		$7$                  & $\approx 0.6786$                       & $\frac{1}{2}(1+\frac{1}{\sqrt{7}}) \approx 0.69$                                                                                              \\ \bottomrule
	\end{tabular}
	\caption{\label{entanglement} The maximal quantum success probability of entanglement assisted one bit communication protocols for $(3,\mathcal{R}_i)$ GRACs, $S^{(3,\mathcal{R}_i)}_{EACC}$, listed along with the number of MUBFs Bob is required to guess. These values are contrasted against the maximal quantum success probability of standard $(n=|\mathcal{R}_i|\rightarrow 1)$ RACs, $S^{(|\mathcal{R}_i|\rightarrow 1)}_{EACC}$, which form lower bounds for the former according to Theorem \ref{GRACLB}. All values were obtained upon coincidence of lower bounds obtained from see-saw semi-definite programming method, and upper bounds obtained \textit{Navascues-Pironio-Acin} hierarchy of semidefinite programming relaxations. For all cases, shared entanglement and one bit of classical communication can achieve a maximal success probability of $\frac{1}{2}\left(1+\frac{1}{\sqrt{|\mathcal{R}_i|}}\right)$.}
\end{table}
It is known that entanglement assistance can increase classical capacity of a quantum channel as established in the seminal super-dense coding protocol \cite{Bennett92} (see also \cite{Bennett99}). Moreover, it has also been shown that entanglement, more generally nonlocal correlations, can increase the {\it zero-error capacity} \cite{Shannon56,Korner98} of a noisy classical channel \cite{Cubitt10,Cubitt11}. More strikingly, as established recently, entanglement can empower even a noiseless classical channel \cite{Frenkel21}.
It is known that EACC protocols can outperform qubit prepare and measure protocols in standard $(n\to 1 )$ RACs when the number of input bits to the sender exceeds three , i.e., for $n>3$ \cite{Pawloski10}. However, as we will report in the following observation, EACC protocols can outperform quantum prepare and measure even with three inputs to the sender in GRACs, 
\begin{observation}
For the case of four MUBFs $\{f_i,f_j,f_k,f_l\}$, such that $\exists \mathbf{x}: f_i(\mathbf{x})\oplus f_j(\mathbf{x}) \neq f_k(\mathbf{x}) \oplus f_l(\mathbf{x})$, entanglement assisted classical communication protocols can outperform the prepare and measure qubit protocols.
\end{observation}      

\section{Discussions and outlook}\label{sec5}
We introduce a generalization of a widely studied family of communication tasks, namely, the random access codes. At this point, it is worth mentioning the recent work \cite{Doriguello21}, wherein the authors also consider a generalization of RACs, referred to as $f$-RACs. In these tasks the receiver intends to recover the value of a given Boolean function $f$ of any subset of fixed size of the sender's input bits. Manifestly, the generalization considered in this work differs from that considered in \cite{Doriguello21}.
The generalization of RACs introduced in this work, namely, GRACs entail the receiver intending to recover certain Boolean functions of the sender's input bits. These functions belong to sets of mutually unbiased functions (MUBS) with the cryptographic property that recovering the value of any one such function does not yield any information about the values of the rest of functions in the set. We study three distinct of classes of protocols for GRACs, (i) classical, (ii) quantum prepare and measure, and (iii) entanglement assisted classical communication protocols. Along with finding general bounds on the success probability of these protocols in GRACs, we have detailed the specific case of GRACs with sender's input data comprising of three independently distributed bits. 

This work motivated several possible directions for future research. While we have studied only classical and quantum strategies, it is also possible to explore more generalized strategies. Note the for axiomatic derivation of Hilbert space quantum mechanics research have initiated the study of generalized probability theories (GPT) \cite{Hardy,Barrett07,Chiribella10,Barnum11,Masanes11}. The seminal two-party-two-input-two-output ($2-2-2$) Popescu-Rohrlich (PR) correlation \cite{Popescu94} that exhibits stronger nonlocal behavior than quantum theory can be studied in this GPT framework.  In \cite{Banik15}, it has been shown that the $(2\to1)$ RAC task can be perfectly accomplished in a particular GPT model called Box world that can be thought as the marginal state space of the set of all $2-2-2$ no-signaling correlations. A particular generalization of this Box world is the polygon model where state spaces are described by symmetric polygons \cite{Janotta11} which has been studied extensively in the recent past \cite{Yopp07,Weis12,Massar14,Safi15,Banik19,Bhattacharya20,Saha20,Saha21}. Performance of theses polygonal models in GRACs is worth exploring.     

Moreover, researchers have generalized the study of RAC/QRAC with larger input-output alphabets, wherein Alice is given random string $\mathbf{x}\equiv x_1\cdots x_n\in\{0,1,\cdots,d-1\}^n$ \cite{Ambainis19,Ambainis15,Tavakoli15}. Indeed, it is possible to generalize MUBF/MUBS and GRACs with larger input-output alphabets. However, we leave this for future study. Finally, as we have demonstrated, GRACs allow for quantum over classical advantage, hence, GRACs may be used for certification of private randomness, and quantum key distribution schemes. It remains to be seen whether GRACs provide an advantage over RACs in such tasks.  

\begin{acknowledgements}
MB acknowledges research grant of INSPIRE-faculty fellowship from the Department of Science and Technology, Government of India. AC acknowledges financial support by NCN grant SHENG 2018/30/Q/ST2/00625. The numerical optimization was carried out using \href{https://ncpol2sdpa.readthedocs.io/en/stable/index.html}{Ncpol2sdpa} \cite{wittek2015algorithm}, \href{https://yalmip.github.io/}{YALMIP} \cite{Lofberg2004} and \href{https://www.mosek.com/documentation/}{SDPT3}~\cite{toh1999sdpt3}.  

\end{acknowledgements}

\clearpage
\section*{Appendix: Explicit qubit prepare and measure protocols for $(n=3,\mathcal{R}_i)$ GRACs}\label{ExplicitAppend}
In this section, we present qubit states and measurements that attain maximal quantum success probability in $(n=3,\mathcal{R}_i)$ GRACs.

{\bf (A)} $|\mathcal{R}_i|=2$: All choices of $\mathcal{R}_i \subset \mathcal{R},$ $(|\mathcal{R}_i|=2)$ are equivalent upto a reordering of the input strings. We give an explicit example using $\mathcal{F}^3_{\mathcal{R}_i}=\{x_1,x_2\}$.  If Bob is asked to evaluate one of the functions from a MUBS of cardinality $2$ then they can have the optimal quantum success by following a strategy similar to the standard $(2\mapsto1)$ RAC. Recall that, optimal quantum protocol for $(2\mapsto1)$ RAC is given by, 
\begin{figure}[b]
	\begin{center}
	\includegraphics[width=0.75\linewidth]{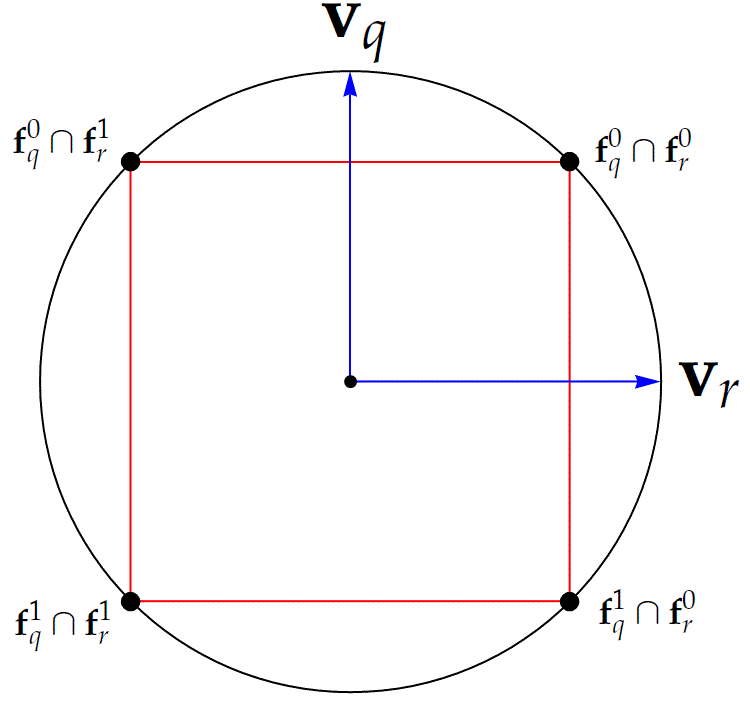}
	\end{center}
	\caption{Optimal quantum protocol for GRAC $(|\mathcal{R}_i|=2)$  Optimal quantum protocol for any $(n=3,\mathcal{R}_i)$ GRAC with $\mathcal{F}^3_{\mathcal{R}_i}\equiv \{f_q,f_r\}$. Shown in figure is a projection of the Bloch-sphere onto a plane with the black dots denoting the encoded states and blue arrows representing the measurement basis. The set of strings $\mathbf{f}^i_q\cap\mathbf{f}^j_r\subset\{0,1\}^3$ are encoded in one of the black dots representing a quantum state, where $\mathbf{f}^i_q$ denote the set of strings whose output value under the function $f_q$ is $i$; $\mathbf{f}^i_q = \{\mathbf{x}\in\{0,1\}^3|f_q(\mathbf{x})=i,i \in \{0,1\}\}$. For evaluating the function $f_q~(f_r)$ von Neumann measurement along $\mathbf{v}_q~(\mathbf{v}_r)$ is performed and function value is reported as $0~(1)$ if the `$+1$'~(`$-1$') outcome clicks.}\label{2rac}
\end{figure}
\begin{align*}
&\mbox{\bf Alice's~encoding:}~\{0,1\}^2\ni x_1x_2\mapsto \rho_{x_1x_2}\\
&\mbox{e.g.}\begin{cases}00\mapsto\frac{1}{2}\left(\mathbf{I}+\frac{1}{\sqrt{2}}\sigma_x+\frac{1}{\sqrt{2}}\sigma_y\right)\\
01\mapsto\frac{1}{2}\left(\mathbf{I}+\frac{1}{\sqrt{2}}\sigma_x-\frac{1}{\sqrt{2}}\sigma_y\right)\\
10\mapsto\frac{1}{2}\left(\mathbf{I}-\frac{1}{\sqrt{2}}\sigma_x+\frac{1}{\sqrt{2}}\sigma_y\right)\\
11\mapsto\frac{1}{2}\left(\mathbf{I}-\frac{1}{\sqrt{2}}\sigma_x-\frac{1}{\sqrt{2}}\sigma_y\right)
\end{cases}
\end{align*}
\begin{align*}
&\mbox{\bf Bob's~decoding:}~M_i\equiv\left\{\frac{1}{2}(\mathbf{I}+\mathbf{v}_i.\sigma),\frac{1}{2}(\mathbf{I}-\mathbf{v}_i.\sigma)\right\}\nonumber\\
&\mbox{e.g.}\begin{cases} 1^{st}~\mbox{function}\to \mathbf{v}_1\equiv(1,0,0)\\
2^{nd}~ \mbox{function}\to \mathbf{v}_2\equiv(0,1,0).\end{cases}
\end{align*}
Bob will guess the bit value as $`0'$ whenever he obtains $`+1'$ outcome, otherwise he guess the value as $`1'$. To make this protocol work for an arbitrary $\mathcal{F}^n_{\mathcal{R}_i}=\{f_q,f_r\}$, Alice follows the encoding $\mathbf{f}^{x_1}_q\cap\mathbf{f}^{x_2}_r\mapsto\ket{\psi}_{x_1x_2}$ and Bob performs the measurement $M_1$ for evaluating the function $f_q$ and performs the measurement $M_2$ for $f_r$ (see Fig.\ref{2rac}). Importantly, in this case both the worst case success probability as well as the average success probability turns out to be $\frac{1}{2}\left(1+\frac{1}{\sqrt{2}}\right)$. 
\\

{\bf (B)} $|\mathcal{R}_i|=3$: Unlike the previous case, all choices of $\mathcal{R}_i \subset \mathcal{R}, (|\mathcal{R}_i|=3)|$ are not equivalent under a permutation of the input strings. If $\mathcal{F}^3_{\mathcal{R}_i}=\{f_q,f_r,f_s\}$, then the two possible equivalence classes are defined by   $f_q\bigoplus f_r \neq f_s$ and $f_q\bigoplus f_r = f_s$.\\

\begin{figure}[t]
\begin{center}
\includegraphics[width=0.75\linewidth]{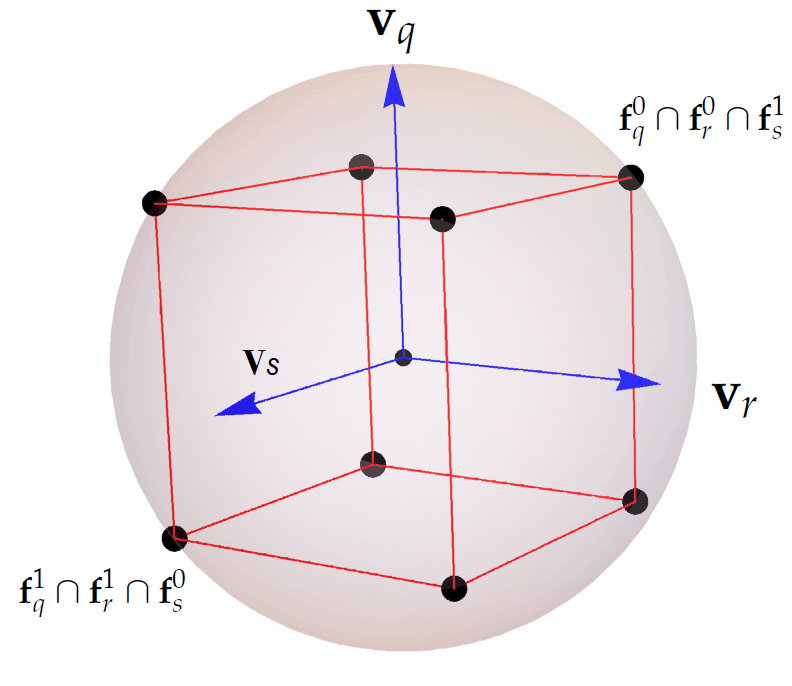}
\end{center}
\caption{Optimal quantum protocol for GRAC $\mathcal{F}^n\{3\}\equiv\{f_q,f_r,f_s\}$, where $f_q\bigoplus f_r \neq f_s$. The set of strings $f^i_q\cap f^j_r\cap f^k_s\subset\{0,1\}^n$ are encoded in qubit state  $\rho_{ijk}:=\frac{1}{2}\left(\mathbf{I}+(-1)^i\frac{1}{\sqrt{3}}\sigma_x+(-1)^j\frac{1}{\sqrt{3}}\sigma_y+(-1)^k\frac{1}{\sqrt{3}}\sigma_z\right)$ where the decoding measurements are along $\mathbf{v}_q\equiv(1,0,0),~\mathbf{v}_r\equiv(0,1,0),~\&~\mathbf{v}_s\equiv(0,0,1)$.}\label{fig2}
\end{figure}

{\bf Case-(i):} If $f_q\bigoplus f_r \neq f_s$, then every such set is equivalent to the set $\mathcal{F}^3_{\mathcal{R}_i}=\{x_1,x_2,x_3\}$. Employing the standard $(3\mapsto1)$ RAC protocol (see Fig. \ref{fig2}) on this set attains the optimal quantum success probability of $\frac{1}{2}\left(1+1/\sqrt{3}\right)$.
\begin{align*}
&\mbox{\bf Alice's~encoding:}\\
&\mbox{e.g.}\begin{cases}x_1x_2x_3\mapsto \rho_{x_1x_2x_3}\\
=\frac{1}{2}\left(\mathbf{I}+(-1)^{x_1}\frac{1}{\sqrt{3}}\sigma_x+(-1)^{x_2}\frac{1}{\sqrt{3}}\sigma_y+(-1)^{x_3}\frac{1}{\sqrt{3}}\sigma_z\right).\end{cases}
\end{align*}
\begin{align*}
&\mbox{\bf Bob's~decoding:}~M_i\equiv\left\{\frac{1}{2}(\mathbf{I}+\mathbf{v}_i.\sigma),\frac{1}{2}(\mathbf{I}-\mathbf{v}_i.\sigma)\right\}\nonumber\\
&\mbox{e.g.}\begin{cases} x_1\to \mathbf{v}_1\equiv(1,0,0)\\
x_2\to \mathbf{v}_2\equiv(0,1,0)\\
x_3\to \mathbf{v}_2\equiv(0,0,1).\end{cases}
\end{align*}

{\bf Case-(ii):} If $f_q\bigoplus f_r = f_s$, then every such set is equivalent to the set $\mathcal{F}^3_{\mathcal{R}_i}=\{x_1,x_2,x_1 \oplus x_2\}$.
\begin{align} \label{32enc}
&\mbox{\bf Alice's~encoding:}\\
&\mbox{e.g.}\begin{cases}\{000,001\} \mapsto\frac{1}{2}\left(\mathbf{I}+\frac{1}{\sqrt{3}}\sigma_x+\frac{1}{\sqrt{3}}\sigma_y+\frac{1}{\sqrt{3}}\sigma_z\right)\\
\{010,011\} \mapsto\frac{1}{2}\left(\mathbf{I}+\frac{1}{\sqrt{3}}\sigma_x-\frac{1}{\sqrt{3}}\sigma_y-\frac{1}{\sqrt{3}}\sigma_z\right)\\
\{100,101\} \mapsto\frac{1}{2}\left(\mathbf{I}-\frac{1}{\sqrt{3}}\sigma_x+\frac{1}{\sqrt{3}}\sigma_y-\frac{1}{\sqrt{3}}\sigma_z\right)\\
\{110,111\} \mapsto\frac{1}{2}\left(\mathbf{I}-\frac{1}{\sqrt{3}}\sigma_x-\frac{1}{\sqrt{3}}\sigma_y+\frac{1}{\sqrt{3}}\sigma_z\right).\end{cases}
\end{align}
\begin{align*}
&\mbox{\bf Bob's~decoding:}~M_i\equiv\left\{\frac{1}{2}(\mathbf{I}+\mathbf{v}_i.\sigma),\frac{1}{2}(\mathbf{I}-\mathbf{v}_i.\sigma)\right\}\nonumber\\
&\mbox{e.g.}\begin{cases} x_1\to \mathbf{v}_1\equiv(1,0,0)\\
x_2\to \mathbf{v}_2\equiv(0,1,0)\\
x_1 \oplus x_2\to \mathbf{v}_{12}\equiv(0,0,1).\end{cases}
\end{align*}
\begin{figure}[h]
\centering
\includegraphics[width=0.75\linewidth]{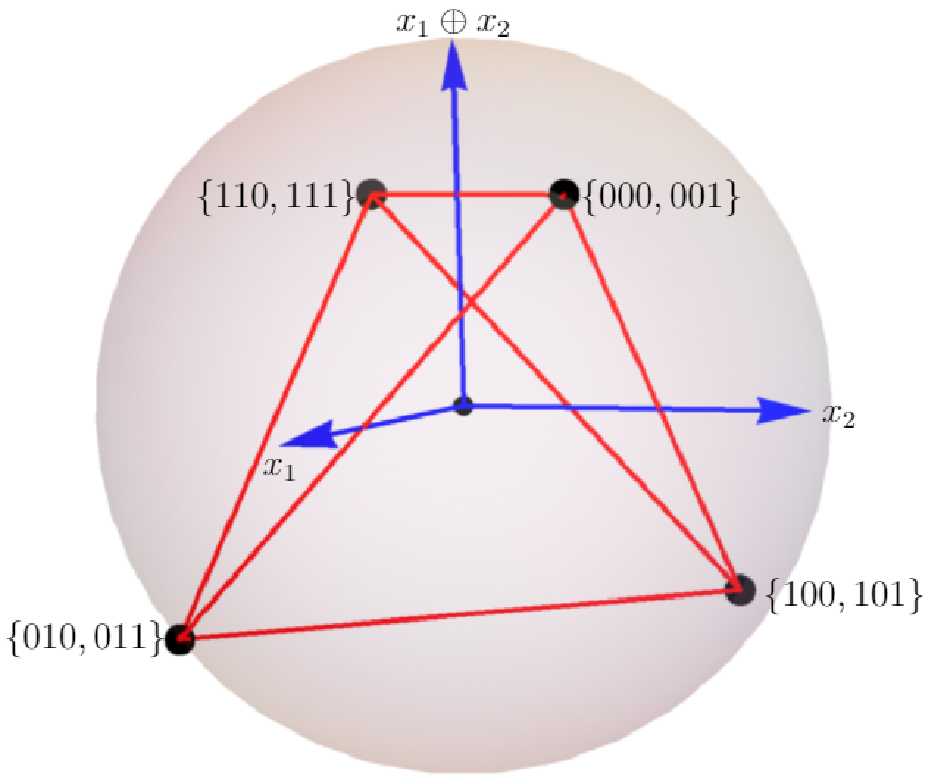}
\caption{($|\mathcal{R}_i|$=$3$) Optimal Encoding protocol for $\mathcal{F}^3_{\mathcal{R}_i}$ = $\{f_q,f_r,f_s\}$ with $f_q\bigoplus f_r = f_s$. Black dots denote the encoded states of Eq.(\ref{32enc}). They form the vertices of a regular tetrahedron. For evaluating the function $x_\alpha\in\{x_1,x_{2},x_1 \oplus x_2\}$, Bob performs the measurement $M_\alpha\equiv\left\{\frac{1}{2}(\mathbf{I}+\mathbf{v}_\alpha.\sigma),\frac{1}{2}(\mathbf{I}-\mathbf{v}_\alpha.\sigma)\right\}$ on the received state and guess the function value as $`0'$ if he obtains outcome $`+1'$, otherwise guess the value as $`1'$. He chooses $\mathbf{v}_1=(1,0,0)=\mathbf{v}_{2}$, and $\mathbf{v}_{12}=(0,0,1)$.}
\label{3rac2}
\end{figure}

{\bf (C)} $\mathcal{F}^3_{\mathcal{R}_i}(|\mathcal{R}_i|=4)$: If $\mathcal{F}^3_{\mathcal{R}_i}=$ $\{f_i,f_j,f_k,f_l\}$, such that $f_i\oplus f_j = f_k \oplus f_l$, then  the optimal classical success is same as the optimal possible quantum success. So in those cases there is no question of quantum advantage. However, when $f_i\oplus f_j \neq f_k \oplus f_l$ optimal classical success is $11/16$ whereas the optimal possible quantum success can go up-to $3/4$. We have find that for these cases the optimal quantum success indeed is higher than the classical value. For instance consider $\mathcal{F}^3_{\mathcal{R}_i}\equiv\{x_1,x_{2},x_{3},x_1 \oplus x_2\}$. 
\begin{figure}[t!]
\centering
\includegraphics[width=0.75\linewidth]{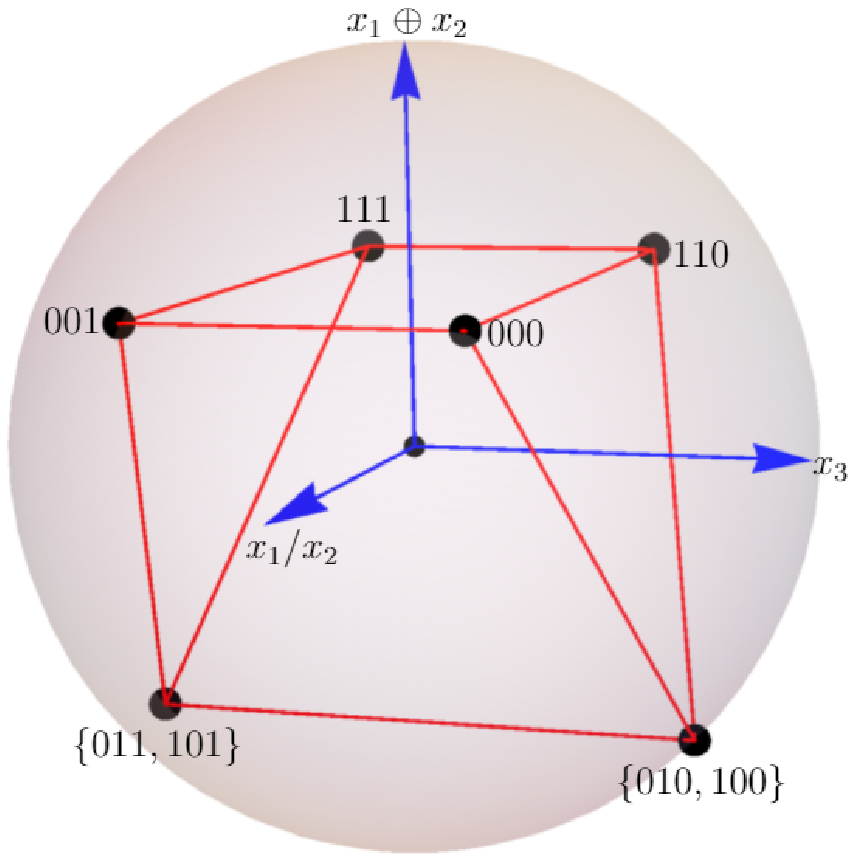}
\caption{$(|\mathcal{R}_i|=4)$ Black dots denote the encoded states of Eq.(\ref{4enc}). For evaluating the function $x_\alpha\in\{x_1,x_{2},x_{3},x_1 \oplus x_2\}$, Bob performs the measurement $M_\alpha\equiv\left\{\frac{1}{2}(\mathbf{I}+\mathbf{v}_\alpha.\sigma),\frac{1}{2}(\mathbf{I}-\mathbf{v}_\alpha.\sigma)\right\}$ on the received state and guess the function value as $`0'$ if he obtains outcome $`+1'$, otherwise guess the value as $`1'$. He chooses $\mathbf{v}_1=(1,0,0)=\mathbf{v}_{2}$, $\mathbf{v}_{3}=(0,1,0)$, and $\mathbf{v}_{12}=(0,0,1)$.}\label{4rac}
\end{figure}
\begin{align} \label{4enc}
&\mbox{\bf Alice's~encoding:}\\
&\mbox{e.g.}\begin{cases}000 \mapsto\frac{1}{2}\left(\mathbf{I}+\sqrt{\frac{2}{3}}\sigma_x+\frac{1}{\sqrt{6}}\sigma_y+\frac{1}{\sqrt{6}}\sigma_z\right)\\
001 \mapsto\frac{1}{2}\left(\mathbf{I}+\sqrt{\frac{2}{3}}\sigma_x-\frac{1}{\sqrt{6}}\sigma_y+\frac{1}{\sqrt{6}}\sigma_z\right)\\
\{010,100\} \mapsto\frac{1}{2}\left(\mathbf{I}+\frac{1}{\sqrt{2}}\sigma_y-\frac{1}{\sqrt{2}}\sigma_z\right)\\
\{011,101\}  \mapsto\frac{1}{2}\left(\mathbf{I}-\frac{1}{\sqrt{2}}\sigma_y-\frac{1}{\sqrt{2}}\sigma_z\right)\\
110\mapsto\frac{1}{2}\left(\mathbf{I}-\sqrt{\frac{2}{3}}\sigma_x+\frac{1}{\sqrt{6}}\sigma_y+\frac{1}{\sqrt{6}}\sigma_z\right)\\
111 \mapsto\frac{1}{2}\left(\mathbf{I}-\sqrt{\frac{2}{3}}\sigma_x-\frac{1}{\sqrt{6}}\sigma_y+\frac{1}{\sqrt{6}}\sigma_z\right).\end{cases}
\end{align}
\begin{align*}
&\mbox{\bf Bob's~decoding:}~M_i\equiv\left\{\frac{1}{2}(\mathbf{I}+\mathbf{v}_i.\sigma),\frac{1}{2}(\mathbf{I}-\mathbf{v}_i.\sigma)\right\}\nonumber\\
&\mbox{e.g.}\begin{cases} x_1\to \mathbf{v}_1\equiv(1,0,0)\\
x_{2}\to \mathbf{v}_{2}\equiv\mathbf{v}_1\\
x_{3}\to \mathbf{v}_{3}\equiv(0,1,0)\\
x_1 \oplus x_2\to \mathbf{v}_{12}\equiv(0,0,1).\end{cases}
\end{align*}
This protocol yields the average success probability $\frac{1}{2}\left(1+\frac{\sqrt{2}+\sqrt{6}}{8}\right)$. Note that the average success is still less than $3/4$. However, upto numerical precision, the lower bound obtained from see-saw semi-definite programming method, and upper bounds obtained via \textit{Navascues-Vertesi} hierarchy of semidefinite programming relaxations is same as the value obtained with the present explicit protocol (see Table \ref{quantum}).   

{\bf (D)} $\boldsymbol{|\mathcal{R}_i|=5:}$ Let us consider a particular case $\mathcal{F}^3_{\mathcal{R}_i}\equiv\{x_1,x_2,x_3,x_1 \oplus x_2,x_1 \oplus x_3\}$. \begin{figure}[t!]
	\begin{center}
\hspace{8pt}\includegraphics[width=0.75\linewidth]{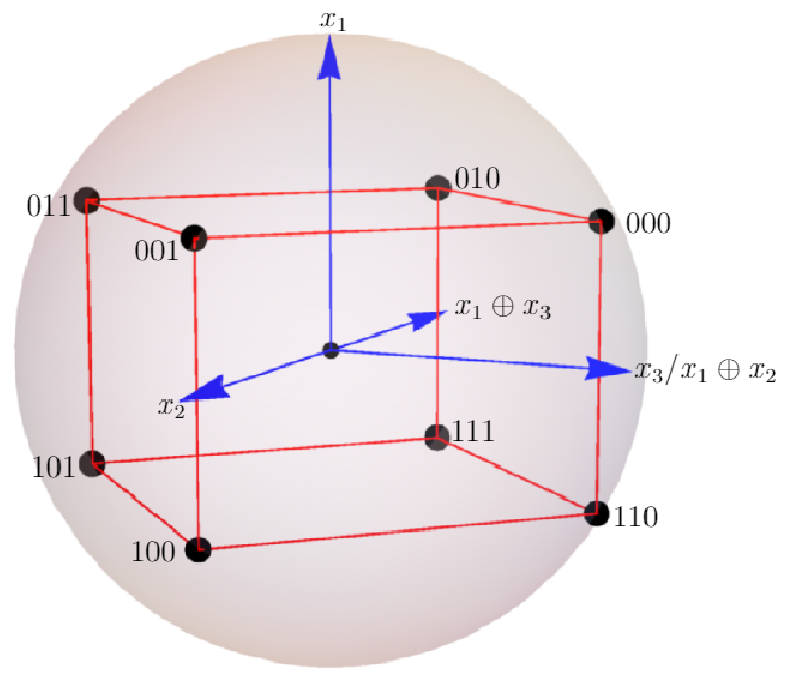}
\end{center}
\caption{$(|\mathcal{R}_i|=5)$ Optimal quantum protocol (non-planar) for a GRAC with $\mathcal{F}^3_{\mathcal{R}_i}\equiv\{x_1,x_2,x_3,x_1 \oplus x_2,x_1 \oplus x_3\}$.}
\label{5rac1} 
\end{figure}
\begin{align*}
&\mbox{\bf Alice's~encoding~(non-planar):}\\
&\mbox{e.g.}\begin{cases}000 \mapsto\frac{1}{2}\left(\mathbf{I}+\frac{1}{\sqrt{5}}\sigma_x+\frac{2}{\sqrt{5}}\sigma_z\right)\\
001 \mapsto\frac{1}{2}\left(\mathbf{I}+\frac{1}{\sqrt{5}}\sigma_x+\frac{2}{\sqrt{5}}\sigma_y\right)\\
010 \mapsto\frac{1}{2}\left(\mathbf{I}+\frac{1}{\sqrt{5}}\sigma_x-\frac{2}{\sqrt{5}}\sigma_y\right)\\
011 \mapsto\frac{1}{2}\left(\mathbf{I}+\frac{1}{\sqrt{5}}\sigma_x-\frac{2}{\sqrt{5}}\sigma_z\right)\\
100 \mapsto\frac{1}{2}\left(\mathbf{I}-\frac{1}{\sqrt{5}}\sigma_x+\frac{2}{\sqrt{5}}\sigma_y\right)\\
101 \mapsto\frac{1}{2}\left(\mathbf{I}-\frac{1}{\sqrt{5}}\sigma_x-\frac{2}{\sqrt{5}}\sigma_z\right)\\
110 \mapsto\frac{1}{2}\left(\mathbf{I}-\frac{1}{\sqrt{5}}\sigma_x+\frac{2}{\sqrt{5}}\sigma_z\right)\\
111 \mapsto\frac{1}{2}\left(\mathbf{I}-\frac{1}{\sqrt{5}}\sigma_x-\frac{2}{\sqrt{5}}\sigma_y\right)\end{cases}
\end{align*}
\begin{align*}
&\mbox{\bf Bob's~decoding:}~M_i\equiv\left\{\frac{1}{2}(\mathbf{I}+\mathbf{v}_i.\sigma),\frac{1}{2}(\mathbf{I}-\mathbf{v}_i.\sigma)\right\}\nonumber\\
&\mbox{e.g.}\begin{cases}
x_1\to \mathbf{v}_1\equiv(1,0,0)\\
x_2\to \mathbf{v}_2\equiv(0,1,0)\\
x_3\to \mathbf{v}_3\equiv(0,0,1)\\
x_1 \oplus x_2\to \mathbf{v}_{12}\equiv\mathbf{v}_3\\
x_1 \oplus x_3\to \mathbf{v}_{13}\equiv-\mathbf{v}_2.\end{cases}
\end{align*}
A straightforward calculation yields the average success probability $\frac{1}{2}\left(1+\frac{1}{\sqrt{5}}\right)$ for this particular encoding-decoding, which turns out to be the optimal quantum success (see Table \ref{quantum}). Note that the encoded states form a rectangular box (see Fig.\ref{fig4-5}). We, however, find that a different strategy where the encoded states lie on a great circle (Fig. \ref{fig5}) but yields the maximum success. 
\begin{figure}[t!]
\begin{center}
		\includegraphics[width=0.75\linewidth]{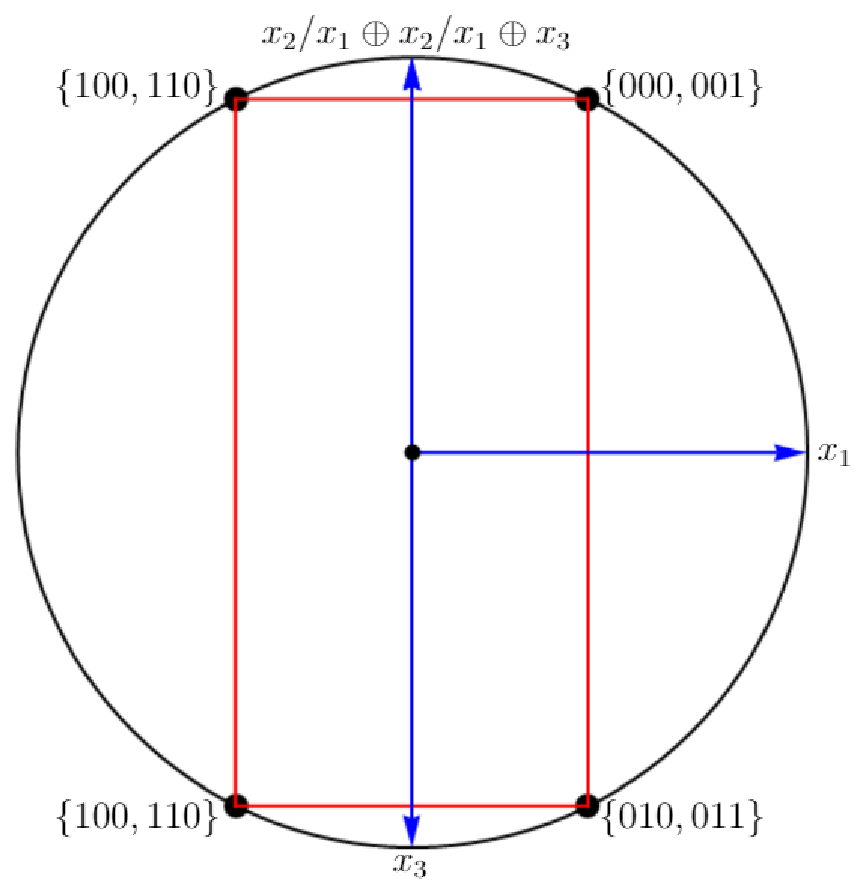}
	\end{center}
	
\caption{$(|\mathcal{R}_i|=5)$ Optimal quantum protocol (planar) for a GRAC with $\mathcal{F}^3_{\mathcal{R}_i}\equiv\{x_1,x_2,x_3,x_1 \oplus x_2,x_1 \oplus x_3\}$.}\label{fig4-5}
	\label{fig5} 
\end{figure}
\begin{align*}
&\mbox{\bf Alice's~encoding~(planar):}\\
&\mbox{e.g.}\begin{cases}
\{000,001\} \mapsto\frac{1}{2}\left(\mathbf{I}+\frac{1}{\sqrt{5}}\sigma_x+\frac{2}{\sqrt{5}}\sigma_y\right)\\
\{010,011\} \mapsto\frac{1}{2}\left(\mathbf{I}+\frac{1}{\sqrt{5}}\sigma_x-\frac{2}{\sqrt{5}}\sigma_y\right)\\
\{101,111\} \mapsto\frac{1}{2}\left(\mathbf{I}-\frac{1}{\sqrt{5}}\sigma_x+\frac{2}{\sqrt{5}}\sigma_y\right)\\
\{100,110\} \mapsto\frac{1}{2}\left(\mathbf{I}-\frac{1}{\sqrt{5}}\sigma_x-\frac{2}{\sqrt{5}}\sigma_y\right).
\end{cases}
\end{align*}
\begin{align*}
&\mbox{\bf Bob's~decoding:}~M_i\equiv\left\{\frac{1}{2}(\mathbf{I}+\mathbf{v}_i.\sigma),\frac{1}{2}(\mathbf{I}-\mathbf{v}_i.\sigma)\right\}\nonumber\\
&\mbox{e.g.}\begin{cases}
x_1\to \mathbf{v}_1\equiv(1,0,0)\\
x_2\to \mathbf{v}_2\equiv(0,1,0)\\
x_3\to \mathbf{v}_3\equiv-\mathbf{v}_2\\
x_1 \oplus x_2\to \mathbf{v}_{12}\equiv\mathbf{v}_2\\
x_1 \oplus x_3\to \mathbf{v}_{13}\equiv\mathbf{v}_2.\end{cases}
\end{align*}

{\bf (E)} $\mathcal{F}^3_{\mathcal{R}_i}(|\mathcal{R}_i|=6)$: Let us consider a particular case $\mathcal{F}^3_{\mathcal{R}_i}\equiv\{x_1,x_2,x_3,x_1 \oplus x_2,x_1 \oplus x_3,x_2 \oplus x_3\}$. 
\begin{figure}[t!]
	\begin{center}
		\includegraphics[scale=0.35]{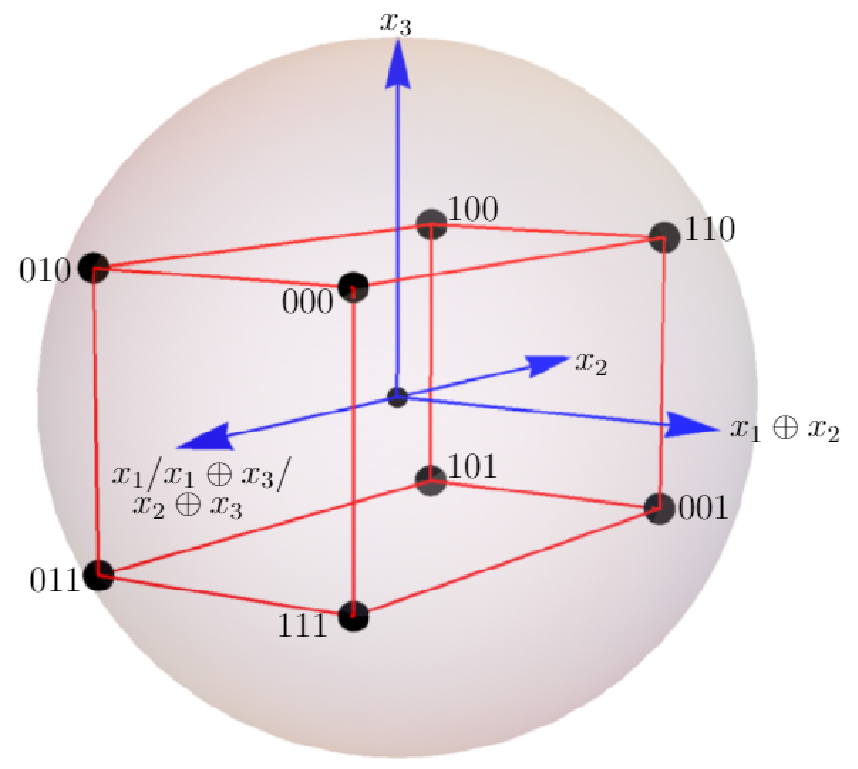}
	\end{center}
	\caption{$(|\mathcal{R}_i|=6)$ Encoded states and decoding measurements corresponding the optimal quantum protocol for $\mathcal{F}^3_{\mathcal{R}_i}\equiv\{x_1,x_2,x_3,x_1 \oplus x_2,x_1 \oplus x_3,x_2 \oplus x_3\}$.}\label{fig6}
\label{6rac} 
\end{figure}
\begin{align*}
&\mbox{\bf Alice's~encoding~(planar):}\\
&\mbox{e.g.}\begin{cases}
000 \mapsto\frac{1}{2}\left(\mathbf{I}+\sqrt{\frac{2}{3}}\sigma_x+\frac{1}{\sqrt{6}}\sigma_y+\frac{1}{\sqrt{6}}\sigma_z\right)\\
001 \mapsto\frac{1}{2}\left(\mathbf{I}-\sqrt{\frac{2}{3}}\sigma_x+\frac{1}{\sqrt{6}}\sigma_y-\frac{1}{\sqrt{6}}\sigma_z\right)\\
010 \mapsto\frac{1}{2}\left(\mathbf{I}+\sqrt{\frac{2}{3}}\sigma_x-\frac{1}{\sqrt{6}}\sigma_y+\frac{1}{\sqrt{6}}\sigma_z\right)\\
011 \mapsto\frac{1}{2}\left(\mathbf{I}+\sqrt{\frac{2}{3}}\sigma_x-\frac{1}{\sqrt{6}}\sigma_y-\frac{1}{\sqrt{6}}\sigma_z\right)\\
100 \mapsto\frac{1}{2}\left(\mathbf{I}-\sqrt{\frac{2}{3}}\sigma_x-\frac{1}{\sqrt{6}}\sigma_y+\frac{1}{\sqrt{6}}\sigma_z\right)\\
101 \mapsto\frac{1}{2}\left(\mathbf{I}-\sqrt{\frac{2}{3}}\sigma_x-\frac{1}{\sqrt{6}}\sigma_y-\frac{1}{\sqrt{6}}\sigma_z\right)\\
110 \mapsto\frac{1}{2}\left(\mathbf{I}-\sqrt{\frac{2}{3}}\sigma_x+\frac{1}{\sqrt{6}}\sigma_y+\frac{1}{\sqrt{6}}\sigma_z\right)\\
111\mapsto\frac{1}{2}\left(\mathbf{I}+\sqrt{\frac{2}{3}}\sigma_x+\frac{1}{\sqrt{6}}\sigma_y-\frac{1}{\sqrt{6}}\sigma_z\right).
\end{cases}
\end{align*}
\begin{align*}
&\mbox{\bf Bob's~decoding:}~M_i\equiv\left\{\frac{1}{2}(\mathbf{I}+\mathbf{v}_i.\sigma),\frac{1}{2}(\mathbf{I}-\mathbf{v}_i.\sigma)\right\}\nonumber\\
&\mbox{e.g.}\begin{cases}
x_1\to \mathbf{v}_1\equiv(1,0,0)\\
x_3\to \mathbf{v}_3\equiv(0,0,1)\\
x_1 \oplus x_2\to \mathbf{v}_{12}\equiv(0,1,0)\\
x_2\to \mathbf{v}_2\equiv-\mathbf{v}_1\\
x_1 \oplus x_3\to \mathbf{v}_{13}\equiv\mathbf{v}_1\\
x_2 \oplus x_3\to \mathbf{v}_{23}\equiv\mathbf{v}_1.\end{cases}
\end{align*}
For this encoding-decoding the average success probability turns out to be $P=\frac{1}{2}\left(1+\frac{1}{\sqrt{6}}\right)$, which is the optimal possible quantum success.

\bibliographystyle{apsrev4-1}
\bibliography{cite}

\begin{thebibliography}{60}%
\makeatletter
\providecommand \@ifxundefined [1]{%
 \@ifx{#1\undefined}
}%
\providecommand \@ifnum [1]{%
 \ifnum #1\expandafter \@firstoftwo
 \else \expandafter \@secondoftwo
 \fi
}%
\providecommand \@ifx [1]{%
 \ifx #1\expandafter \@firstoftwo
 \else \expandafter \@secondoftwo
 \fi
}%
\providecommand \natexlab [1]{#1}%
\providecommand \enquote  [1]{``#1''}%
\providecommand \bibnamefont  [1]{#1}%
\providecommand \bibfnamefont [1]{#1}%
\providecommand \citenamefont [1]{#1}%
\providecommand \href@noop [0]{\@secondoftwo}%
\providecommand \href [0]{\begingroup \@sanitize@url \@href}%
\providecommand \@href[1]{\@@startlink{#1}\@@href}%
\providecommand \@@href[1]{\endgroup#1\@@endlink}%
\providecommand \@sanitize@url [0]{\catcode `\\12\catcode `\$12\catcode
  `\&12\catcode `\#12\catcode `\^12\catcode `\_12\catcode `\%12\relax}%
\providecommand \@@startlink[1]{}%
\providecommand \@@endlink[0]{}%
\providecommand \url  [0]{\begingroup\@sanitize@url \@url }%
\providecommand \@url [1]{\endgroup\@href {#1}{\urlprefix }}%
\providecommand \urlprefix  [0]{URL }%
\providecommand \Eprint [0]{\href }%
\providecommand \doibase [0]{http://dx.doi.org/}%
\providecommand \selectlanguage [0]{\@gobble}%
\providecommand \bibinfo  [0]{\@secondoftwo}%
\providecommand \bibfield  [0]{\@secondoftwo}%
\providecommand \translation [1]{[#1]}%
\providecommand \BibitemOpen [0]{}%
\providecommand \bibitemStop [0]{}%
\providecommand \bibitemNoStop [0]{.\EOS\space}%
\providecommand \EOS [0]{\spacefactor3000\relax}%
\providecommand \BibitemShut  [1]{\csname bibitem#1\endcsname}%
\let\auto@bib@innerbib\@empty
\bibitem [{\citenamefont {Shannon}(1948)}]{Shannon48}%
  \BibitemOpen
  \bibfield  {author} {\bibinfo {author} {\bibfnamefont {C.~E.}\ \bibnamefont
  {Shannon}},\ }\href {\doibase 10.1002/j.1538-7305.1948.tb01338.x} {\bibfield
  {journal} {\bibinfo  {journal} {The Bell System Technical Journal}\ }\textbf
  {\bibinfo {volume} {27}},\ \bibinfo {pages} {379} (\bibinfo {year}
  {1948})}\BibitemShut {NoStop}%
\bibitem [{\citenamefont {Bennett}\ and\ \citenamefont
  {Wiesner}(1992)}]{Bennett92}%
  \BibitemOpen
  \bibfield  {author} {\bibinfo {author} {\bibfnamefont {C.~H.}\ \bibnamefont
  {Bennett}}\ and\ \bibinfo {author} {\bibfnamefont {S.~J.}\ \bibnamefont
  {Wiesner}},\ }\href {\doibase 10.1103/PhysRevLett.69.2881} {\bibfield
  {journal} {\bibinfo  {journal} {Phys. Rev. Lett.}\ }\textbf {\bibinfo
  {volume} {69}},\ \bibinfo {pages} {2881} (\bibinfo {year}
  {1992})}\BibitemShut {NoStop}%
\bibitem [{\citenamefont {Holevo}(1973)}]{Holevo73}%
  \BibitemOpen
  \bibfield  {author} {\bibinfo {author} {\bibfnamefont {A.~S.}\ \bibnamefont
  {Holevo}},\ }\href {http://mi.mathnet.ru/eng/ppi903} {\bibfield  {journal}
  {\bibinfo  {journal} {Problemy Peredachi Informatsii}\ }\textbf {\bibinfo
  {volume} {9}},\ \bibinfo {pages} {3} (\bibinfo {year} {1973})}\BibitemShut
  {NoStop}%
\bibitem [{\citenamefont {Frenkel}\ and\ \citenamefont
  {Weiner}(2015)}]{Frenkel15}%
  \BibitemOpen
  \bibfield  {author} {\bibinfo {author} {\bibfnamefont {P.~E.}\ \bibnamefont
  {Frenkel}}\ and\ \bibinfo {author} {\bibfnamefont {M.}~\bibnamefont
  {Weiner}},\ }\href {\doibase 10.1007/s00220-015-2463-0} {\bibfield  {journal}
  {\bibinfo  {journal} {Communications in Mathematical Physics}\ }\textbf
  {\bibinfo {volume} {340}},\ \bibinfo {pages} {563} (\bibinfo {year}
  {2015})}\BibitemShut {NoStop}%
\bibitem [{\citenamefont {Wiesner}(1983)}]{Wiesner83}%
  \BibitemOpen
  \bibfield  {author} {\bibinfo {author} {\bibfnamefont {S.}~\bibnamefont
  {Wiesner}},\ }\href {\doibase https://doi.org/10.1145/1008908.1008920}
  {\bibfield  {journal} {\bibinfo  {journal} {ACM Sigact News}\ }\textbf
  {\bibinfo {volume} {15}},\ \bibinfo {pages} {78} (\bibinfo {year}
  {1983})}\BibitemShut {NoStop}%
\bibitem [{\citenamefont {Ambainis}\ \emph {et~al.}(1999)\citenamefont
  {Ambainis}, \citenamefont {Nayak}, \citenamefont {Ta-Shma},\ and\
  \citenamefont {Vazirani}}]{Ambainis99}%
  \BibitemOpen
  \bibfield  {author} {\bibinfo {author} {\bibfnamefont {A.}~\bibnamefont
  {Ambainis}}, \bibinfo {author} {\bibfnamefont {A.}~\bibnamefont {Nayak}},
  \bibinfo {author} {\bibfnamefont {A.}~\bibnamefont {Ta-Shma}}, \ and\
  \bibinfo {author} {\bibfnamefont {U.}~\bibnamefont {Vazirani}},\ }in\ \href
  {\doibase 10.1145/301250.301347} {\emph {\bibinfo {booktitle} {Proceedings of
  the thirty-first annual ACM symposium on Theory of Computing}}}\ (\bibinfo
  {year} {1999})\ pp.\ \bibinfo {pages} {376--383}\BibitemShut {NoStop}%
\bibitem [{\citenamefont {Ambainis}\ \emph {et~al.}(2002)\citenamefont
  {Ambainis}, \citenamefont {Nayak}, \citenamefont {Ta-Shma},\ and\
  \citenamefont {Vazirani}}]{Ambainis02}%
  \BibitemOpen
  \bibfield  {author} {\bibinfo {author} {\bibfnamefont {A.}~\bibnamefont
  {Ambainis}}, \bibinfo {author} {\bibfnamefont {A.}~\bibnamefont {Nayak}},
  \bibinfo {author} {\bibfnamefont {A.}~\bibnamefont {Ta-Shma}}, \ and\
  \bibinfo {author} {\bibfnamefont {U.}~\bibnamefont {Vazirani}},\ }\href
  {\doibase 10.1145/581771.581773} {\bibfield  {journal} {\bibinfo  {journal}
  {J. ACM}\ }\textbf {\bibinfo {volume} {49}},\ \bibinfo {pages} {496–511}
  (\bibinfo {year} {2002})}\BibitemShut {NoStop}%
\bibitem [{\citenamefont {Buhrman}\ \emph {et~al.}(2010)\citenamefont
  {Buhrman}, \citenamefont {Cleve}, \citenamefont {Massar},\ and\ \citenamefont
  {de~Wolf}}]{Buhrman10}%
  \BibitemOpen
  \bibfield  {author} {\bibinfo {author} {\bibfnamefont {H.}~\bibnamefont
  {Buhrman}}, \bibinfo {author} {\bibfnamefont {R.}~\bibnamefont {Cleve}},
  \bibinfo {author} {\bibfnamefont {S.}~\bibnamefont {Massar}}, \ and\ \bibinfo
  {author} {\bibfnamefont {R.}~\bibnamefont {de~Wolf}},\ }\href {\doibase
  10.1103/RevModPhys.82.665} {\bibfield  {journal} {\bibinfo  {journal} {Rev.
  Mod. Phys.}\ }\textbf {\bibinfo {volume} {82}},\ \bibinfo {pages} {665}
  (\bibinfo {year} {2010})}\BibitemShut {NoStop}%
\bibitem [{\citenamefont {Klauck}(2001)}]{Klauck01}%
  \BibitemOpen
  \bibfield  {author} {\bibinfo {author} {\bibfnamefont {H.}~\bibnamefont
  {Klauck}},\ }in\ \href {\doibase 10.1109/SFCS.2001.959903} {\emph {\bibinfo
  {booktitle} {Proceedings 42nd IEEE Symposium on Foundations of Computer
  Science}}}\ (\bibinfo {organization} {IEEE},\ \bibinfo {year} {2001})\ pp.\
  \bibinfo {pages} {288--297}\BibitemShut {NoStop}%
\bibitem [{\citenamefont {Kerenidis}\ and\ \citenamefont
  {De~Wolf}(2004)}]{Kerenidis04}%
  \BibitemOpen
  \bibfield  {author} {\bibinfo {author} {\bibfnamefont {I.}~\bibnamefont
  {Kerenidis}}\ and\ \bibinfo {author} {\bibfnamefont {R.}~\bibnamefont
  {De~Wolf}},\ }\href
  {https://www.sciencedirect.com/science/article/pii/S0022000004000467}
  {\bibfield  {journal} {\bibinfo  {journal} {Journal of Computer and System
  Sciences}\ }\textbf {\bibinfo {volume} {69}},\ \bibinfo {pages} {395}
  (\bibinfo {year} {2004})}\BibitemShut {NoStop}%
\bibitem [{\citenamefont {Aaronson}(2004)}]{Aaronso04}%
  \BibitemOpen
  \bibfield  {author} {\bibinfo {author} {\bibfnamefont {S.}~\bibnamefont
  {Aaronson}},\ }in\ \href {\doibase 10.1109/CCC.2004.1313854} {\emph {\bibinfo
  {booktitle} {Proceedings. 19th IEEE Annual Conference on Computational
  Complexity, 2004.}}}\ (\bibinfo {organization} {IEEE},\ \bibinfo {year}
  {2004})\ pp.\ \bibinfo {pages} {320--332}\BibitemShut {NoStop}%
\bibitem [{\citenamefont {Wehner}\ and\ \citenamefont
  {De~Wolf}(2005)}]{Wehner05}%
  \BibitemOpen
  \bibfield  {author} {\bibinfo {author} {\bibfnamefont {S.}~\bibnamefont
  {Wehner}}\ and\ \bibinfo {author} {\bibfnamefont {R.}~\bibnamefont
  {De~Wolf}},\ }in\ \href {\doibase https://doi.org/10.1007/11523468_115}
  {\emph {\bibinfo {booktitle} {International Colloquium on Automata,
  Languages, and Programming}}}\ (\bibinfo {organization} {Springer},\ \bibinfo
  {year} {2005})\ pp.\ \bibinfo {pages} {1424--1436}\BibitemShut {NoStop}%
\bibitem [{\citenamefont {Gavinsky}\ \emph {et~al.}(2009)\citenamefont
  {Gavinsky}, \citenamefont {Kempe}, \citenamefont {Regev},\ and\ \citenamefont
  {De~Wolf}}]{Gavinsky06}%
  \BibitemOpen
  \bibfield  {author} {\bibinfo {author} {\bibfnamefont {D.}~\bibnamefont
  {Gavinsky}}, \bibinfo {author} {\bibfnamefont {J.}~\bibnamefont {Kempe}},
  \bibinfo {author} {\bibfnamefont {O.}~\bibnamefont {Regev}}, \ and\ \bibinfo
  {author} {\bibfnamefont {R.}~\bibnamefont {De~Wolf}},\ }\href {\doibase
  https://doi.org/10.1137/060665798} {\bibfield  {journal} {\bibinfo  {journal}
  {SIAM Journal on Computing}\ }\textbf {\bibinfo {volume} {39}},\ \bibinfo
  {pages} {1} (\bibinfo {year} {2009})}\BibitemShut {NoStop}%
\bibitem [{\citenamefont {Hayashi}\ \emph {et~al.}(2007)\citenamefont
  {Hayashi}, \citenamefont {Iwama}, \citenamefont {Nishimura}, \citenamefont
  {Raymond},\ and\ \citenamefont {Yamashita}}]{Hayashi07}%
  \BibitemOpen
  \bibfield  {author} {\bibinfo {author} {\bibfnamefont {M.}~\bibnamefont
  {Hayashi}}, \bibinfo {author} {\bibfnamefont {K.}~\bibnamefont {Iwama}},
  \bibinfo {author} {\bibfnamefont {H.}~\bibnamefont {Nishimura}}, \bibinfo
  {author} {\bibfnamefont {R.}~\bibnamefont {Raymond}}, \ and\ \bibinfo
  {author} {\bibfnamefont {S.}~\bibnamefont {Yamashita}},\ }in\ \href {\doibase
  https://doi.org/10.1007/978-3-540-70918-3_52} {\emph {\bibinfo {booktitle}
  {Annual Symposium on Theoretical Aspects of Computer Science}}}\ (\bibinfo
  {organization} {Springer},\ \bibinfo {year} {2007})\ pp.\ \bibinfo {pages}
  {610--621}\BibitemShut {NoStop}%
\bibitem [{\citenamefont {Spekkens}\ \emph {et~al.}(2009)\citenamefont
  {Spekkens}, \citenamefont {Buzacott}, \citenamefont {Keehn}, \citenamefont
  {Toner},\ and\ \citenamefont {Pryde}}]{Spekkens09}%
  \BibitemOpen
  \bibfield  {author} {\bibinfo {author} {\bibfnamefont {R.~W.}\ \bibnamefont
  {Spekkens}}, \bibinfo {author} {\bibfnamefont {D.~H.}\ \bibnamefont
  {Buzacott}}, \bibinfo {author} {\bibfnamefont {A.~J.}\ \bibnamefont {Keehn}},
  \bibinfo {author} {\bibfnamefont {B.}~\bibnamefont {Toner}}, \ and\ \bibinfo
  {author} {\bibfnamefont {G.~J.}\ \bibnamefont {Pryde}},\ }\href {\doibase
  https://doi.org/10.1103/PhysRevLett.102.010401} {\bibfield  {journal}
  {\bibinfo  {journal} {Physical review letters}\ }\textbf {\bibinfo {volume}
  {102}},\ \bibinfo {pages} {010401} (\bibinfo {year} {2009})}\BibitemShut
  {NoStop}%
\bibitem [{\citenamefont {Banik}\ \emph {et~al.}(2015)\citenamefont {Banik},
  \citenamefont {Bhattacharya}, \citenamefont {Mukherjee}, \citenamefont {Roy},
  \citenamefont {Ambainis},\ and\ \citenamefont {Rai}}]{Banik15}%
  \BibitemOpen
  \bibfield  {author} {\bibinfo {author} {\bibfnamefont {M.}~\bibnamefont
  {Banik}}, \bibinfo {author} {\bibfnamefont {S.~S.}\ \bibnamefont
  {Bhattacharya}}, \bibinfo {author} {\bibfnamefont {A.}~\bibnamefont
  {Mukherjee}}, \bibinfo {author} {\bibfnamefont {A.}~\bibnamefont {Roy}},
  \bibinfo {author} {\bibfnamefont {A.}~\bibnamefont {Ambainis}}, \ and\
  \bibinfo {author} {\bibfnamefont {A.}~\bibnamefont {Rai}},\ }\href {\doibase
  https://doi.org/10.1103/PhysRevA.92.030103} {\bibfield  {journal} {\bibinfo
  {journal} {Physical Review A}\ }\textbf {\bibinfo {volume} {92}},\ \bibinfo
  {pages} {030103} (\bibinfo {year} {2015})}\BibitemShut {NoStop}%
\bibitem [{\citenamefont {Chailloux}\ \emph {et~al.}(2016)\citenamefont
  {Chailloux}, \citenamefont {Kerenidis}, \citenamefont {Kundu},\ and\
  \citenamefont {Sikora}}]{Chailloux16}%
  \BibitemOpen
  \bibfield  {author} {\bibinfo {author} {\bibfnamefont {A.}~\bibnamefont
  {Chailloux}}, \bibinfo {author} {\bibfnamefont {I.}~\bibnamefont
  {Kerenidis}}, \bibinfo {author} {\bibfnamefont {S.}~\bibnamefont {Kundu}}, \
  and\ \bibinfo {author} {\bibfnamefont {J.}~\bibnamefont {Sikora}},\ }\href
  {\doibase 10.1088/1367-2630/18/4/045003} {\bibfield  {journal} {\bibinfo
  {journal} {New Journal of Physics}\ }\textbf {\bibinfo {volume} {18}},\
  \bibinfo {pages} {045003} (\bibinfo {year} {2016})}\BibitemShut {NoStop}%
\bibitem [{\citenamefont {Hameedi}\ \emph {et~al.}(2017)\citenamefont
  {Hameedi}, \citenamefont {Tavakoli}, \citenamefont {Marques},\ and\
  \citenamefont {Bourennane}}]{Hameedi17}%
  \BibitemOpen
  \bibfield  {author} {\bibinfo {author} {\bibfnamefont {A.}~\bibnamefont
  {Hameedi}}, \bibinfo {author} {\bibfnamefont {A.}~\bibnamefont {Tavakoli}},
  \bibinfo {author} {\bibfnamefont {B.}~\bibnamefont {Marques}}, \ and\
  \bibinfo {author} {\bibfnamefont {M.}~\bibnamefont {Bourennane}},\ }\href
  {\doibase https://doi.org/10.1103/PhysRevLett.119.220402} {\bibfield
  {journal} {\bibinfo  {journal} {Physical review letters}\ }\textbf {\bibinfo
  {volume} {119}},\ \bibinfo {pages} {220402} (\bibinfo {year}
  {2017})}\BibitemShut {NoStop}%
\bibitem [{\citenamefont {Ghorai}\ and\ \citenamefont {Pan}(2018)}]{Ghorai18}%
  \BibitemOpen
  \bibfield  {author} {\bibinfo {author} {\bibfnamefont {S.}~\bibnamefont
  {Ghorai}}\ and\ \bibinfo {author} {\bibfnamefont {A.}~\bibnamefont {Pan}},\
  }\href {\doibase https://doi.org/10.1103/PhysRevA.98.032110} {\bibfield
  {journal} {\bibinfo  {journal} {Physical Review A}\ }\textbf {\bibinfo
  {volume} {98}},\ \bibinfo {pages} {032110} (\bibinfo {year}
  {2018})}\BibitemShut {NoStop}%
\bibitem [{\citenamefont {Saha}\ and\ \citenamefont
  {Chaturvedi}(2019)}]{Saha19}%
  \BibitemOpen
  \bibfield  {author} {\bibinfo {author} {\bibfnamefont {D.}~\bibnamefont
  {Saha}}\ and\ \bibinfo {author} {\bibfnamefont {A.}~\bibnamefont
  {Chaturvedi}},\ }\href {\doibase https://doi.org/10.1103/PhysRevA.100.022108}
  {\bibfield  {journal} {\bibinfo  {journal} {Physical Review A}\ }\textbf
  {\bibinfo {volume} {100}},\ \bibinfo {pages} {022108} (\bibinfo {year}
  {2019})}\BibitemShut {NoStop}%
\bibitem [{\citenamefont {Ambainis}\ \emph {et~al.}(2019)\citenamefont
  {Ambainis}, \citenamefont {Banik}, \citenamefont {Chaturvedi}, \citenamefont
  {Kravchenko},\ and\ \citenamefont {Rai}}]{Ambainis19}%
  \BibitemOpen
  \bibfield  {author} {\bibinfo {author} {\bibfnamefont {A.}~\bibnamefont
  {Ambainis}}, \bibinfo {author} {\bibfnamefont {M.}~\bibnamefont {Banik}},
  \bibinfo {author} {\bibfnamefont {A.}~\bibnamefont {Chaturvedi}}, \bibinfo
  {author} {\bibfnamefont {D.}~\bibnamefont {Kravchenko}}, \ and\ \bibinfo
  {author} {\bibfnamefont {A.}~\bibnamefont {Rai}},\ }\href {\doibase
  https://doi.org/10.1007/s11128-019-2228-3} {\bibfield  {journal} {\bibinfo
  {journal} {Quantum Information Processing}\ }\textbf {\bibinfo {volume}
  {18}},\ \bibinfo {pages} {1} (\bibinfo {year} {2019})}\BibitemShut {NoStop}%
\bibitem [{\citenamefont {Chaturvedi}\ and\ \citenamefont
  {Saha}(2020)}]{Chaturvedi20}%
  \BibitemOpen
  \bibfield  {author} {\bibinfo {author} {\bibfnamefont {A.}~\bibnamefont
  {Chaturvedi}}\ and\ \bibinfo {author} {\bibfnamefont {D.}~\bibnamefont
  {Saha}},\ }\href {\doibase https://doi.org/10.22331/q-2020-10-21-345}
  {\bibfield  {journal} {\bibinfo  {journal} {Quantum}\ }\textbf {\bibinfo
  {volume} {4}},\ \bibinfo {pages} {345} (\bibinfo {year} {2020})}\BibitemShut
  {NoStop}%
\bibitem [{\citenamefont {Paw{\l}owski}\ and\ \citenamefont
  {{\.Z}ukowski}(2010)}]{Pawloski10}%
  \BibitemOpen
  \bibfield  {author} {\bibinfo {author} {\bibfnamefont {M.}~\bibnamefont
  {Paw{\l}owski}}\ and\ \bibinfo {author} {\bibfnamefont {M.}~\bibnamefont
  {{\.Z}ukowski}},\ }\href {\doibase
  https://doi.org/10.1103/PhysRevA.81.042326} {\bibfield  {journal} {\bibinfo
  {journal} {Physical Review A}\ }\textbf {\bibinfo {volume} {81}},\ \bibinfo
  {pages} {042326} (\bibinfo {year} {2010})}\BibitemShut {NoStop}%
\bibitem [{\citenamefont {Chaturvedi}\ \emph {et~al.}(2017)\citenamefont
  {Chaturvedi}, \citenamefont {Pawlowski},\ and\ \citenamefont
  {Horodecki}}]{AnubhavRAC}%
  \BibitemOpen
  \bibfield  {author} {\bibinfo {author} {\bibfnamefont {A.}~\bibnamefont
  {Chaturvedi}}, \bibinfo {author} {\bibfnamefont {M.}~\bibnamefont
  {Pawlowski}}, \ and\ \bibinfo {author} {\bibfnamefont {K.}~\bibnamefont
  {Horodecki}},\ }\href {\doibase 10.1103/PhysRevA.96.022125} {\bibfield
  {journal} {\bibinfo  {journal} {Phys. Rev. A}\ }\textbf {\bibinfo {volume}
  {96}},\ \bibinfo {pages} {022125} (\bibinfo {year} {2017})}\BibitemShut
  {NoStop}%
\bibitem [{\citenamefont {Paw{\l}owski}\ \emph {et~al.}(2009)\citenamefont
  {Paw{\l}owski}, \citenamefont {Paterek}, \citenamefont {Kaszlikowski},
  \citenamefont {Scarani}, \citenamefont {Winter},\ and\ \citenamefont
  {{\.Z}ukowski}}]{Pawloski09}%
  \BibitemOpen
  \bibfield  {author} {\bibinfo {author} {\bibfnamefont {M.}~\bibnamefont
  {Paw{\l}owski}}, \bibinfo {author} {\bibfnamefont {T.}~\bibnamefont
  {Paterek}}, \bibinfo {author} {\bibfnamefont {D.}~\bibnamefont
  {Kaszlikowski}}, \bibinfo {author} {\bibfnamefont {V.}~\bibnamefont
  {Scarani}}, \bibinfo {author} {\bibfnamefont {A.}~\bibnamefont {Winter}}, \
  and\ \bibinfo {author} {\bibfnamefont {M.}~\bibnamefont {{\.Z}ukowski}},\
  }\href {\doibase https://doi.org/10.1038/nature08400} {\bibfield  {journal}
  {\bibinfo  {journal} {Nature}\ }\textbf {\bibinfo {volume} {461}},\ \bibinfo
  {pages} {1101} (\bibinfo {year} {2009})}\BibitemShut {NoStop}%
\bibitem [{\citenamefont {Al-Safi}\ and\ \citenamefont {Short}(2011)}]{Safi11}%
  \BibitemOpen
  \bibfield  {author} {\bibinfo {author} {\bibfnamefont {S.~W.}\ \bibnamefont
  {Al-Safi}}\ and\ \bibinfo {author} {\bibfnamefont {A.~J.}\ \bibnamefont
  {Short}},\ }\href {\doibase https://doi.org/10.1103/PhysRevA.84.042323}
  {\bibfield  {journal} {\bibinfo  {journal} {Physical Review A}\ }\textbf
  {\bibinfo {volume} {84}},\ \bibinfo {pages} {042323} (\bibinfo {year}
  {2011})}\BibitemShut {NoStop}%
\bibitem [{\citenamefont {Paw\l{}owski}\ and\ \citenamefont
  {Brunner}(2011)}]{Marcin2011}%
  \BibitemOpen
  \bibfield  {author} {\bibinfo {author} {\bibfnamefont {M.}~\bibnamefont
  {Paw\l{}owski}}\ and\ \bibinfo {author} {\bibfnamefont {N.}~\bibnamefont
  {Brunner}},\ }\href {\doibase 10.1103/PhysRevA.84.010302} {\bibfield
  {journal} {\bibinfo  {journal} {Phys. Rev. A}\ }\textbf {\bibinfo {volume}
  {84}},\ \bibinfo {pages} {010302} (\bibinfo {year} {2011})}\BibitemShut
  {NoStop}%
\bibitem [{\citenamefont {Chaturvedi}\ \emph {et~al.}(2018)\citenamefont
  {Chaturvedi}, \citenamefont {Ray}, \citenamefont {Veynar},\ and\
  \citenamefont {Paw{\l}owski}}]{Chaturvedi2018}%
  \BibitemOpen
  \bibfield  {author} {\bibinfo {author} {\bibfnamefont {A.}~\bibnamefont
  {Chaturvedi}}, \bibinfo {author} {\bibfnamefont {M.}~\bibnamefont {Ray}},
  \bibinfo {author} {\bibfnamefont {R.}~\bibnamefont {Veynar}}, \ and\ \bibinfo
  {author} {\bibfnamefont {M.}~\bibnamefont {Paw{\l}owski}},\ }\href
  {https://link.springer.com/article/10.1007/s11128-018-1892-z} {\bibfield
  {journal} {\bibinfo  {journal} {Quantum information processing}\ }\textbf
  {\bibinfo {volume} {17}},\ \bibinfo {pages} {1} (\bibinfo {year}
  {2018})}\BibitemShut {NoStop}%
\bibitem [{\citenamefont {Chaturvedi}\ \emph {et~al.}(2020)\citenamefont
  {Chaturvedi}, \citenamefont {Farkas},\ and\ \citenamefont
  {Wright}}]{Anubhav2020}%
  \BibitemOpen
  \bibfield  {author} {\bibinfo {author} {\bibfnamefont {A.}~\bibnamefont
  {Chaturvedi}}, \bibinfo {author} {\bibfnamefont {M.}~\bibnamefont {Farkas}},
  \ and\ \bibinfo {author} {\bibfnamefont {V.~J.}\ \bibnamefont {Wright}},\
  }\href {https://arxiv.org/abs/2010.05853} {\bibfield  {journal} {\bibinfo
  {journal} {arXiv:2010.05853}\ } (\bibinfo {year} {2020})}\BibitemShut
  {NoStop}%
\bibitem [{\citenamefont {Ambainis}\ \emph {et~al.}(2008)\citenamefont
  {Ambainis}, \citenamefont {Leung}, \citenamefont {Mancinska},\ and\
  \citenamefont {Ozols}}]{Ambainis08}%
  \BibitemOpen
  \bibfield  {author} {\bibinfo {author} {\bibfnamefont {A.}~\bibnamefont
  {Ambainis}}, \bibinfo {author} {\bibfnamefont {D.}~\bibnamefont {Leung}},
  \bibinfo {author} {\bibfnamefont {L.}~\bibnamefont {Mancinska}}, \ and\
  \bibinfo {author} {\bibfnamefont {M.}~\bibnamefont {Ozols}},\ }\href
  {https://arxiv.org/abs/0810.2937v3} {\bibfield  {journal} {\bibinfo
  {journal} {arXiv preprint arXiv:0810.2937}\ } (\bibinfo {year}
  {2008})}\BibitemShut {NoStop}%
\bibitem [{\citenamefont {Ambainis}\ \emph {et~al.}(2015)\citenamefont
  {Ambainis}, \citenamefont {Kravchenko},\ and\ \citenamefont
  {Rai}}]{Ambainis15}%
  \BibitemOpen
  \bibfield  {author} {\bibinfo {author} {\bibfnamefont {A.}~\bibnamefont
  {Ambainis}}, \bibinfo {author} {\bibfnamefont {D.}~\bibnamefont
  {Kravchenko}}, \ and\ \bibinfo {author} {\bibfnamefont {A.}~\bibnamefont
  {Rai}},\ }\href {https://arxiv.org/abs/1510.03045v1} {\bibfield  {journal}
  {\bibinfo  {journal} {arXiv preprint arXiv:1510.03045}\ } (\bibinfo {year}
  {2015})}\BibitemShut {NoStop}%
\bibitem [{\citenamefont {Kraus}\ \emph {et~al.}(1983)\citenamefont {Kraus},
  \citenamefont {B{\"o}hm}, \citenamefont {Dollard},\ and\ \citenamefont
  {Wootters}}]{Kraus83}%
  \BibitemOpen
  \bibfield  {author} {\bibinfo {author} {\bibfnamefont {K.}~\bibnamefont
  {Kraus}}, \bibinfo {author} {\bibfnamefont {A.}~\bibnamefont {B{\"o}hm}},
  \bibinfo {author} {\bibfnamefont {J.~D.}\ \bibnamefont {Dollard}}, \ and\
  \bibinfo {author} {\bibfnamefont {W.}~\bibnamefont {Wootters}},\ }\href
  {https://pascal-francis.inist.fr/vibad/index.php?action=getRecordDetail&idt=9385051}
  {\bibfield  {journal} {\bibinfo  {journal} {Lecture notes in physics}\
  }\textbf {\bibinfo {volume} {190}} (\bibinfo {year} {1983})}\BibitemShut
  {NoStop}%
\bibitem [{\citenamefont {Nielsen}\ and\ \citenamefont
  {Chuang}(2002)}]{Chuang00}%
  \BibitemOpen
  \bibfield  {author} {\bibinfo {author} {\bibfnamefont {M.~A.}\ \bibnamefont
  {Nielsen}}\ and\ \bibinfo {author} {\bibfnamefont {I.}~\bibnamefont
  {Chuang}},\ }\href
  {https://aapt.scitation.org/doi/abs/10.1119/1.1463744?journalCode=ajp}
  {\enquote {\bibinfo {title} {Quantum computation and quantum information},}\
  } (\bibinfo {year} {2002})\BibitemShut {NoStop}%
\bibitem [{\citenamefont {Wilde}(2013)}]{Wilde13}%
  \BibitemOpen
  \bibfield  {author} {\bibinfo {author} {\bibfnamefont {M.~M.}\ \bibnamefont
  {Wilde}},\ }\href
  {https://books.google.pl/books?hl=pl&lr=&id=dphmVbPSLHMC&oi=fnd&pg=PR11&dq=Mark+M.+Wilde%3B+Quantum+Information+Theory&ots=UdP2DQp03O&sig=e2P48c32XQH4McjsSihcDx5vVqE&redir_esc=y#v=onepage&q=Mark%20M.%20Wilde%3B%20Quantum%20Information%20Theory&f=false}
  {\emph {\bibinfo {title} {Quantum information theory}}}\ (\bibinfo
  {publisher} {Cambridge University Press},\ \bibinfo {year}
  {2013})\BibitemShut {NoStop}%
\bibitem [{\citenamefont {Bennett}\ \emph {et~al.}(1999)\citenamefont
  {Bennett}, \citenamefont {Shor}, \citenamefont {Smolin},\ and\ \citenamefont
  {Thapliyal}}]{Bennett99}%
  \BibitemOpen
  \bibfield  {author} {\bibinfo {author} {\bibfnamefont {C.~H.}\ \bibnamefont
  {Bennett}}, \bibinfo {author} {\bibfnamefont {P.~W.}\ \bibnamefont {Shor}},
  \bibinfo {author} {\bibfnamefont {J.~A.}\ \bibnamefont {Smolin}}, \ and\
  \bibinfo {author} {\bibfnamefont {A.~V.}\ \bibnamefont {Thapliyal}},\ }\href
  {\doibase 10.1103/PhysRevLett.83.3081} {\bibfield  {journal} {\bibinfo
  {journal} {Phys. Rev. Lett.}\ }\textbf {\bibinfo {volume} {83}},\ \bibinfo
  {pages} {3081} (\bibinfo {year} {1999})}\BibitemShut {NoStop}%
\bibitem [{\citenamefont {Shannon}(1956)}]{Shannon56}%
  \BibitemOpen
  \bibfield  {author} {\bibinfo {author} {\bibfnamefont {C.}~\bibnamefont
  {Shannon}},\ }\href {\doibase 10.1109/TIT.1956.1056798} {\bibfield  {journal}
  {\bibinfo  {journal} {IRE Transactions on Information Theory}\ }\textbf
  {\bibinfo {volume} {2}},\ \bibinfo {pages} {8} (\bibinfo {year}
  {1956})}\BibitemShut {NoStop}%
\bibitem [{\citenamefont {Korner}\ and\ \citenamefont
  {Orlitsky}(1998)}]{Korner98}%
  \BibitemOpen
  \bibfield  {author} {\bibinfo {author} {\bibfnamefont {J.}~\bibnamefont
  {Korner}}\ and\ \bibinfo {author} {\bibfnamefont {A.}~\bibnamefont
  {Orlitsky}},\ }\href {\doibase 10.1109/18.720537} {\bibfield  {journal}
  {\bibinfo  {journal} {IEEE Transactions on Information Theory}\ }\textbf
  {\bibinfo {volume} {44}},\ \bibinfo {pages} {2207} (\bibinfo {year}
  {1998})}\BibitemShut {NoStop}%
\bibitem [{\citenamefont {Cubitt}\ \emph {et~al.}(2010)\citenamefont {Cubitt},
  \citenamefont {Leung}, \citenamefont {Matthews},\ and\ \citenamefont
  {Winter}}]{Cubitt10}%
  \BibitemOpen
  \bibfield  {author} {\bibinfo {author} {\bibfnamefont {T.~S.}\ \bibnamefont
  {Cubitt}}, \bibinfo {author} {\bibfnamefont {D.}~\bibnamefont {Leung}},
  \bibinfo {author} {\bibfnamefont {W.}~\bibnamefont {Matthews}}, \ and\
  \bibinfo {author} {\bibfnamefont {A.}~\bibnamefont {Winter}},\ }\href
  {\doibase 10.1103/PhysRevLett.104.230503} {\bibfield  {journal} {\bibinfo
  {journal} {Phys. Rev. Lett.}\ }\textbf {\bibinfo {volume} {104}},\ \bibinfo
  {pages} {230503} (\bibinfo {year} {2010})}\BibitemShut {NoStop}%
\bibitem [{\citenamefont {Cubitt}\ \emph {et~al.}(2011)\citenamefont {Cubitt},
  \citenamefont {Leung}, \citenamefont {Matthews},\ and\ \citenamefont
  {Winter}}]{Cubitt11}%
  \BibitemOpen
  \bibfield  {author} {\bibinfo {author} {\bibfnamefont {T.~S.}\ \bibnamefont
  {Cubitt}}, \bibinfo {author} {\bibfnamefont {D.}~\bibnamefont {Leung}},
  \bibinfo {author} {\bibfnamefont {W.}~\bibnamefont {Matthews}}, \ and\
  \bibinfo {author} {\bibfnamefont {A.}~\bibnamefont {Winter}},\ }\href
  {\doibase 10.1109/TIT.2011.2159047} {\bibfield  {journal} {\bibinfo
  {journal} {IEEE Transactions on Information Theory}\ }\textbf {\bibinfo
  {volume} {57}},\ \bibinfo {pages} {5509} (\bibinfo {year}
  {2011})}\BibitemShut {NoStop}%
\bibitem [{\citenamefont {Frenkel}\ and\ \citenamefont
  {Weiner}(2021)}]{Frenkel21}%
  \BibitemOpen
  \bibfield  {author} {\bibinfo {author} {\bibfnamefont {P.~E.}\ \bibnamefont
  {Frenkel}}\ and\ \bibinfo {author} {\bibfnamefont {M.}~\bibnamefont
  {Weiner}},\ }\href {https://arxiv.org/abs/2103.08567} {\bibfield  {journal}
  {\bibinfo  {journal} {arXiv:2103.08567}\ } (\bibinfo {year}
  {2021})}\BibitemShut {NoStop}%
\bibitem [{\citenamefont {Doriguello}\ and\ \citenamefont
  {Montanaro}(2021)}]{Doriguello21}%
  \BibitemOpen
  \bibfield  {author} {\bibinfo {author} {\bibfnamefont {J.~F.}\ \bibnamefont
  {Doriguello}}\ and\ \bibinfo {author} {\bibfnamefont {A.}~\bibnamefont
  {Montanaro}},\ }\href {\doibase 10.22331/q-2021-03-07-402} {\bibfield
  {journal} {\bibinfo  {journal} {{Quantum}}\ }\textbf {\bibinfo {volume}
  {5}},\ \bibinfo {pages} {402} (\bibinfo {year} {2021})}\BibitemShut {NoStop}%
\bibitem [{\citenamefont {Hardy}(2001)}]{Hardy}%
  \BibitemOpen
  \bibfield  {author} {\bibinfo {author} {\bibfnamefont {L.}~\bibnamefont
  {Hardy}},\ }\href {https://arxiv.org/abs/quant-ph/0101012} {\bibfield
  {journal} {\bibinfo  {journal} {arXiv preprint quant-ph/0101012}\ } (\bibinfo
  {year} {2001})}\BibitemShut {NoStop}%
\bibitem [{\citenamefont {Barrett}(2007)}]{Barrett07}%
  \BibitemOpen
  \bibfield  {author} {\bibinfo {author} {\bibfnamefont {J.}~\bibnamefont
  {Barrett}},\ }\href {\doibase 10.1103/PhysRevA.75.032304} {\bibfield
  {journal} {\bibinfo  {journal} {Phys. Rev. A}\ }\textbf {\bibinfo {volume}
  {75}},\ \bibinfo {pages} {032304} (\bibinfo {year} {2007})}\BibitemShut
  {NoStop}%
\bibitem [{\citenamefont {Chiribella}\ \emph {et~al.}(2010)\citenamefont
  {Chiribella}, \citenamefont {D'Ariano},\ and\ \citenamefont
  {Perinotti}}]{Chiribella10}%
  \BibitemOpen
  \bibfield  {author} {\bibinfo {author} {\bibfnamefont {G.}~\bibnamefont
  {Chiribella}}, \bibinfo {author} {\bibfnamefont {G.~M.}\ \bibnamefont
  {D'Ariano}}, \ and\ \bibinfo {author} {\bibfnamefont {P.}~\bibnamefont
  {Perinotti}},\ }\href {\doibase 10.1103/PhysRevA.81.062348} {\bibfield
  {journal} {\bibinfo  {journal} {Phys. Rev. A}\ }\textbf {\bibinfo {volume}
  {81}},\ \bibinfo {pages} {062348} (\bibinfo {year} {2010})}\BibitemShut
  {NoStop}%
\bibitem [{\citenamefont {Barnum}\ and\ \citenamefont
  {Wilce}(2011)}]{Barnum11}%
  \BibitemOpen
  \bibfield  {author} {\bibinfo {author} {\bibfnamefont {H.}~\bibnamefont
  {Barnum}}\ and\ \bibinfo {author} {\bibfnamefont {A.}~\bibnamefont {Wilce}},\
  }\href {\doibase https://doi.org/10.1016/j.entcs.2011.01.002} {\bibfield
  {journal} {\bibinfo  {journal} {Electronic Notes in Theoretical Computer
  Science}\ }\textbf {\bibinfo {volume} {270}},\ \bibinfo {pages} {3} (\bibinfo
  {year} {2011})},\ \bibinfo {note} {proceedings of the Joint 5th International
  Workshop on Quantum Physics and Logic and 4th Workshop on Developments in
  Computational Models (QPL/DCM 2008)}\BibitemShut {NoStop}%
\bibitem [{\citenamefont {Masanes}\ and\ \citenamefont
  {Müller}(2011)}]{Masanes11}%
  \BibitemOpen
  \bibfield  {author} {\bibinfo {author} {\bibfnamefont {L.}~\bibnamefont
  {Masanes}}\ and\ \bibinfo {author} {\bibfnamefont {M.~P.}\ \bibnamefont
  {Müller}},\ }\href {\doibase 10.1088/1367-2630/13/6/063001} {\bibfield
  {journal} {\bibinfo  {journal} {New Journal of Physics}\ }\textbf {\bibinfo
  {volume} {13}},\ \bibinfo {pages} {063001} (\bibinfo {year}
  {2011})}\BibitemShut {NoStop}%
\bibitem [{\citenamefont {{Popescu}}\ and\ \citenamefont
  {{Rohrlich}}(1994)}]{Popescu94}%
  \BibitemOpen
  \bibfield  {author} {\bibinfo {author} {\bibfnamefont {S.}~\bibnamefont
  {{Popescu}}}\ and\ \bibinfo {author} {\bibfnamefont {D.}~\bibnamefont
  {{Rohrlich}}},\ }\href {https://link.springer.com/article/10.1007} {\bibfield
   {journal} {\bibinfo  {journal} {Foundations of Physics}\ }\textbf {\bibinfo
  {volume} {24}},\ \bibinfo {pages} {379} (\bibinfo {year} {1994})}\BibitemShut
  {NoStop}%
\bibitem [{\citenamefont {Janotta}\ \emph {et~al.}(2011)\citenamefont
  {Janotta}, \citenamefont {Gogolin}, \citenamefont {Barrett},\ and\
  \citenamefont {Brunner}}]{Janotta11}%
  \BibitemOpen
  \bibfield  {author} {\bibinfo {author} {\bibfnamefont {P.}~\bibnamefont
  {Janotta}}, \bibinfo {author} {\bibfnamefont {C.}~\bibnamefont {Gogolin}},
  \bibinfo {author} {\bibfnamefont {J.}~\bibnamefont {Barrett}}, \ and\
  \bibinfo {author} {\bibfnamefont {N.}~\bibnamefont {Brunner}},\ }\href
  {\doibase 10.1088/1367-2630/13/6/063024} {\bibfield  {journal} {\bibinfo
  {journal} {New Journal of Physics}\ }\textbf {\bibinfo {volume} {13}},\
  \bibinfo {pages} {063024} (\bibinfo {year} {2011})}\BibitemShut {NoStop}%
\bibitem [{\citenamefont {Yopp}\ and\ \citenamefont
  {Hill~{\S}}(2005)}]{Yopp07}%
  \BibitemOpen
  \bibfield  {author} {\bibinfo {author} {\bibfnamefont {D.~A.}\ \bibnamefont
  {Yopp}}\ and\ \bibinfo {author} {\bibfnamefont {R.~D.}\ \bibnamefont
  {Hill~{\S}}},\ }\href {\doibase https://doi.org/10.1080/03081080412331272500}
  {\bibfield  {journal} {\bibinfo  {journal} {Linear and Multilinear Algebra}\
  }\textbf {\bibinfo {volume} {53}},\ \bibinfo {pages} {167} (\bibinfo {year}
  {2005})}\BibitemShut {NoStop}%
\bibitem [{\citenamefont {Weis}(2011)}]{Weis12}%
  \BibitemOpen
  \bibfield  {author} {\bibinfo {author} {\bibfnamefont {S.}~\bibnamefont
  {Weis}},\ }\href {https://arxiv.org/abs/1107.2319v2} {\bibfield  {journal}
  {\bibinfo  {journal} {arXiv preprint arXiv:1107.2319}\ } (\bibinfo {year}
  {2011})}\BibitemShut {NoStop}%
\bibitem [{\citenamefont {Massar}\ and\ \citenamefont
  {Patra}(2014)}]{Massar14}%
  \BibitemOpen
  \bibfield  {author} {\bibinfo {author} {\bibfnamefont {S.}~\bibnamefont
  {Massar}}\ and\ \bibinfo {author} {\bibfnamefont {M.~K.}\ \bibnamefont
  {Patra}},\ }\href {\doibase https://doi.org/10.1103/PhysRevA.89.052124}
  {\bibfield  {journal} {\bibinfo  {journal} {Physical Review A}\ }\textbf
  {\bibinfo {volume} {89}},\ \bibinfo {pages} {052124} (\bibinfo {year}
  {2014})}\BibitemShut {NoStop}%
\bibitem [{\citenamefont {Al-Safi}\ and\ \citenamefont
  {Richens}(2015)}]{Safi15}%
  \BibitemOpen
  \bibfield  {author} {\bibinfo {author} {\bibfnamefont {S.~W.}\ \bibnamefont
  {Al-Safi}}\ and\ \bibinfo {author} {\bibfnamefont {J.}~\bibnamefont
  {Richens}},\ }\href {\doibase 10.1088/1367-2630/17/12/123001} {\bibfield
  {journal} {\bibinfo  {journal} {New Journal of Physics}\ }\textbf {\bibinfo
  {volume} {17}},\ \bibinfo {pages} {123001} (\bibinfo {year}
  {2015})}\BibitemShut {NoStop}%
\bibitem [{\citenamefont {Banik}\ \emph {et~al.}(2019)\citenamefont {Banik},
  \citenamefont {Saha}, \citenamefont {Guha}, \citenamefont {Agrawal},
  \citenamefont {Bhattacharya}, \citenamefont {Roy},\ and\ \citenamefont
  {Majumdar}}]{Banik19}%
  \BibitemOpen
  \bibfield  {author} {\bibinfo {author} {\bibfnamefont {M.}~\bibnamefont
  {Banik}}, \bibinfo {author} {\bibfnamefont {S.}~\bibnamefont {Saha}},
  \bibinfo {author} {\bibfnamefont {T.}~\bibnamefont {Guha}}, \bibinfo {author}
  {\bibfnamefont {S.}~\bibnamefont {Agrawal}}, \bibinfo {author} {\bibfnamefont
  {S.~S.}\ \bibnamefont {Bhattacharya}}, \bibinfo {author} {\bibfnamefont
  {A.}~\bibnamefont {Roy}}, \ and\ \bibinfo {author} {\bibfnamefont
  {A.}~\bibnamefont {Majumdar}},\ }\href {\doibase
  https://doi.org/10.1103/PhysRevA.100.060101} {\bibfield  {journal} {\bibinfo
  {journal} {Physical Review A}\ }\textbf {\bibinfo {volume} {100}},\ \bibinfo
  {pages} {060101} (\bibinfo {year} {2019})}\BibitemShut {NoStop}%
\bibitem [{\citenamefont {Bhattacharya}\ \emph {et~al.}(2020)\citenamefont
  {Bhattacharya}, \citenamefont {Saha}, \citenamefont {Guha},\ and\
  \citenamefont {Banik}}]{Bhattacharya20}%
  \BibitemOpen
  \bibfield  {author} {\bibinfo {author} {\bibfnamefont {S.~S.}\ \bibnamefont
  {Bhattacharya}}, \bibinfo {author} {\bibfnamefont {S.}~\bibnamefont {Saha}},
  \bibinfo {author} {\bibfnamefont {T.}~\bibnamefont {Guha}}, \ and\ \bibinfo
  {author} {\bibfnamefont {M.}~\bibnamefont {Banik}},\ }\href {\doibase
  https://doi.org/10.1103/PhysRevResearch.2.012068} {\bibfield  {journal}
  {\bibinfo  {journal} {Physical Review Research}\ }\textbf {\bibinfo {volume}
  {2}},\ \bibinfo {pages} {012068} (\bibinfo {year} {2020})}\BibitemShut
  {NoStop}%
\bibitem [{\citenamefont {Saha}\ \emph
  {et~al.}(2020{\natexlab{a}})\citenamefont {Saha}, \citenamefont
  {Bhattacharya}, \citenamefont {Guha}, \citenamefont {Halder},\ and\
  \citenamefont {Banik}}]{Saha20}%
  \BibitemOpen
  \bibfield  {author} {\bibinfo {author} {\bibfnamefont {S.}~\bibnamefont
  {Saha}}, \bibinfo {author} {\bibfnamefont {S.~S.}\ \bibnamefont
  {Bhattacharya}}, \bibinfo {author} {\bibfnamefont {T.}~\bibnamefont {Guha}},
  \bibinfo {author} {\bibfnamefont {S.}~\bibnamefont {Halder}}, \ and\ \bibinfo
  {author} {\bibfnamefont {M.}~\bibnamefont {Banik}},\ }\href {\doibase
  https://doi.org/10.1002/andp.202000334} {\bibfield  {journal} {\bibinfo
  {journal} {Annalen der Physik}\ }\textbf {\bibinfo {volume} {532}},\ \bibinfo
  {pages} {2000334} (\bibinfo {year} {2020}{\natexlab{a}})}\BibitemShut
  {NoStop}%
\bibitem [{\citenamefont {Saha}\ \emph
  {et~al.}(2020{\natexlab{b}})\citenamefont {Saha}, \citenamefont {Guha},
  \citenamefont {Bhattacharya},\ and\ \citenamefont {Banik}}]{Saha21}%
  \BibitemOpen
  \bibfield  {author} {\bibinfo {author} {\bibfnamefont {S.}~\bibnamefont
  {Saha}}, \bibinfo {author} {\bibfnamefont {T.}~\bibnamefont {Guha}}, \bibinfo
  {author} {\bibfnamefont {S.~S.}\ \bibnamefont {Bhattacharya}}, \ and\
  \bibinfo {author} {\bibfnamefont {M.}~\bibnamefont {Banik}},\ }\href
  {https://arxiv.org/abs/2012.05781v1} {\bibfield  {journal} {\bibinfo
  {journal} {arXiv preprint arXiv:2012.05781}\ } (\bibinfo {year}
  {2020}{\natexlab{b}})}\BibitemShut {NoStop}%
\bibitem [{\citenamefont {Tavakoli}\ \emph {et~al.}(2015)\citenamefont
  {Tavakoli}, \citenamefont {Hameedi}, \citenamefont {Marques},\ and\
  \citenamefont {Bourennane}}]{Tavakoli15}%
  \BibitemOpen
  \bibfield  {author} {\bibinfo {author} {\bibfnamefont {A.}~\bibnamefont
  {Tavakoli}}, \bibinfo {author} {\bibfnamefont {A.}~\bibnamefont {Hameedi}},
  \bibinfo {author} {\bibfnamefont {B.}~\bibnamefont {Marques}}, \ and\
  \bibinfo {author} {\bibfnamefont {M.}~\bibnamefont {Bourennane}},\ }\href
  {\doibase 10.1103/PhysRevLett.114.170502} {\bibfield  {journal} {\bibinfo
  {journal} {Phys. Rev. Lett.}\ }\textbf {\bibinfo {volume} {114}},\ \bibinfo
  {pages} {170502} (\bibinfo {year} {2015})}\BibitemShut {NoStop}%
\bibitem [{\citenamefont {Wittek}(2015)}]{wittek2015algorithm}%
  \BibitemOpen
  \bibfield  {author} {\bibinfo {author} {\bibfnamefont {P.}~\bibnamefont
  {Wittek}},\ }\href {https://dl.acm.org/doi/10.1145/2699464} {\bibfield
  {journal} {\bibinfo  {journal} {ACM Trans.~Math.~Softw.}\ }\textbf {\bibinfo
  {volume} {41}},\ \bibinfo {pages} {1} (\bibinfo {year} {2015})}\BibitemShut
  {NoStop}%
\bibitem [{\citenamefont {L{\"{o}}fberg}(2004)}]{Lofberg2004}%
  \BibitemOpen
  \bibfield  {author} {\bibinfo {author} {\bibfnamefont {J.}~\bibnamefont
  {L{\"{o}}fberg}},\ }in\ \href {https://doi.org/10.1109/CACSD.2004.1393890}
  {\emph {\bibinfo {booktitle} {In Proceedings of the CACSD Conference}}}\
  (\bibinfo {address} {Taipei, Taiwan},\ \bibinfo {year} {2004})\BibitemShut
  {NoStop}%
\bibitem [{\citenamefont {Toh}\ \emph {et~al.}(1999)\citenamefont {Toh},
  \citenamefont {Todd},\ and\ \citenamefont
  {T{\"u}t{\"u}nc{\"u}}}]{toh1999sdpt3}%
  \BibitemOpen
  \bibfield  {author} {\bibinfo {author} {\bibfnamefont {K.-C.}\ \bibnamefont
  {Toh}}, \bibinfo {author} {\bibfnamefont {M.~J.}\ \bibnamefont {Todd}}, \
  and\ \bibinfo {author} {\bibfnamefont {R.~H.}\ \bibnamefont
  {T{\"u}t{\"u}nc{\"u}}},\ }\href {https://doi.org/10.1080/10556789908805762}
  {\bibfield  {journal} {\bibinfo  {journal} {Optim.~Methods Softw.~}\ }\textbf
  {\bibinfo {volume} {11}},\ \bibinfo {pages} {545} (\bibinfo {year}
  {1999})}\BibitemShut {NoStop}%
\end{thebibliography}%

\end{document}